\documentclass[11pt,a4paper]{article}

\usepackage{cite}
\usepackage{amsmath,amssymb}
\usepackage{amsthm}
\usepackage{microtype}
\usepackage{framed}
\usepackage[top=1in, bottom=1in, left=1in, right=1in]{geometry}

\usepackage{array}
\usepackage{graphicx}
\usepackage{float}
\usepackage{url}
\usepackage{hyperref}
\usepackage{multirow}
\usepackage{color,soul}
\usepackage{xspace}
\usepackage{tabularx}
\usepackage{subcaption}

\usepackage{enumitem}
\usepackage{algorithm}
\usepackage{framed}
\usepackage[noend]{algpseudocode}

\algblockdefx{AtState}{END}[1]{\underline{In state {\sf #1}:}}{}
\algblockdefx{AtStateBreakLine}{END}[2]{ \underline{In state {\sf #1}} \newline  \underline{{\sf #2}:} }{}
\algcblockdefx{AtState}{AtState}{END}[1]{\underline{In state {\sf #1}:}}{\vspace{-1.5em}}{}
\algcblockdefx{AtStateBreakLine}{AtStateBreakLine}{END}[2]{ \underline{In state {\sf #1}} \newline  \underline{{\sf #2}:} }{}

\newcommand{\Explore}{\textsc{Explore}\xspace}

\newcommand{\sReached}{\textsf{ReachedElected}\xspace}
\newcommand{\sReaching}{\textsf{ReachingElected}\xspace}

\newcommand{\pMeeting}{\textsl{meeting}\xspace}
\newcommand{\pSameDirMeeting}{\textsl{meetingSameDir}\xspace}
\newcommand{\pOppositeDirMeeting}{\textsl{meetingOppositeDir}\xspace}
\newcommand{\pCross}{\textsl{crossed}\xspace}

\newcommand{\dLeft}{\textit{left}\xspace}
\newcommand{\dRight}{\textit{right}\xspace}

\newcommand{\vTtime}{\ensuremath{Ttime}\xspace}
\newcommand{\vOthers}{\ensuremath{Agents}\xspace}
\newcommand{\vTotOthers}{\ensuremath{TotalAgents}\xspace}

\newcommand{\vEtime}{\ensuremath{Etime}\xspace}
\newcommand{\vEsteps}{\ensuremath{Esteps}\xspace}
\newcommand{\vBtime}{\ensuremath{Btime}\xspace}
\newcommand{\vBPtime}{\ensuremath{BPeriods}\xspace}

\newcommand{\predC}{\emph{Pred}\xspace}

\newcommand{\augmented}{\emph{Logic Ring}\xspace}

\newcommand{\Gathering}{{\sc Gathering}\xspace}

\newtheorem{observation}{Observation}
\newtheorem{property}{Property}

\newtheorem{lemma}{Lemma}

\newtheorem{theorem}{Theorem}

\begin{document}

%
\title{Gathering in  Dynamic Rings}


\author{Giuseppe Antonio Di Luna\footnotemark[1], \and Paola Flocchini\footnotemark[1], \and Linda Pagli\footnotemark[2],\and Giuseppe Prencipe\footnotemark[2],\and Nicola Santoro\footnotemark[3],\and Giovanni Viglietta\footnotemark[1]}
\def\thefootnote{\fnsymbol{footnote}}
 \footnotetext[1]{\noindent
 School Electrical Engineering and Computer Science, University of Ottawa,
 Canada.}
  \footnotetext[2]{\noindent
  Dipartimento di Informatica, University of Pisa, Italy}
\footnotetext[3]{\noindent
 School of Computer Science, Carleton University, 
Canada.}

\date{}
\maketitle

\begin{abstract}

The    {\em gathering} (or  {\em multi-agent rendezvous}) problem requires a
set  of mobile agents, Êarbitrarily positioned at different  nodes of a network
 to group within finite time at the same  location,
 not fixed in advanced. 
 
 The extensive  existing literature on this problem 
shares the same fundamental assumption:   the  topological structure does not change during the rendezvous or the gathering;
this is true also for those investigations   that  consider faulty nodes.
In other words, they only consider {\em static graphs}.

In this paper we start the investigation of  gathering in {\em dynamic graphs},
that is  networks where 
 the topology  changes continuously and 
 at unpredictable locations.

We study the feasibility of  gathering  mobile  agents, identical and without explicit communication capabilities,   in
 a {\em dynamic ring}
of  anonymous nodes;
the class of dynamics we consider is the
classic  {\em 1-interval-connectivity}.
We focus on  
 the impact that   factors such as    {\em chirality} (i.e., a  common sense of orientation) and
  {\em cross detection} (i.e., the ability to detect, when traversing an edge, whether some agent is traversing it in the other direction), 
   have on the solvability of the problem; and
we establish several results. 

We  provide
 a complete characterization of the classes of initial
configurations from which  the gathering problem is solvable in presence and in absence of cross detection
and of chirality.
The feasibility results of the characterization 
 are all constructive: we provide distributed algorithms that allow 
the agents to gather  within low polynomial time. In particular, the protocols  for gathering with cross detection
are time {optimal}.  

We  also  show  that  cross detection is a
 powerful computational element. We prove that,
   without chirality,  
 knowledge of the ring size is strictly more powerful 
than knowledge of the number of agents;
 on the other hand, with chirality, knowledge of n can be substituted by knowledge of k, 
 yielding the same classes of feasible initial configurations.

From our investigation it follows that, for the gathering problem,
 the computational
obstacles created by the dynamic nature of the ring
can be overcome by the presence of chirality or of cross-detection.

 \end{abstract}
\newpage

\section{Introduction}

\subsection{Background and Problem}

 The  {\em gathering} problem 
 requires a set of $k$ mobile computational entities, dispersed at different
 locations in the spacial universe they inhabit, to group within finite time at the same  location,
 not fixed in advanced. 
 This problem
  models
 many situations that arise in the
 real world, e.g., searching for or regrouping animals, people,
 equipment, and vehicles, 
 
This problem, known also as {\em multi-agent rendezvous}, has been intesively and extensively studied 
in a variety of fields, including operations research
(e.g.,   \cite{AlG03}) and  control (e.g., \cite{LiMA07}),
the original focus  being  on the    {\em rendezvous} problem, i.e. the special case $k=2$.

In distributed computing, this problem has been extensively studied both in continuous 
and in discrete domains. In the {\em continuous} case, both  the  gathering and
the rendevous problems have  been investigated in the context of swarms of autonomous 
mobile {\em robots}  operating in one- 
 and  two-dimensional spaces,  requiring them to meet at (or converge to)  the same point 
 (e.g., see \cite{CiFPS12,CoP04,Deg+2011,FlPSW05,FlSVY16,PaPV15}).
 
In the {\em discrete} case, the mobile entities, usually called {\em agents}, 
are dispersed in a network modeled as a graph and are  required  to gather at the same node (or at the two sides of the same edge)
and terminate (e.g., see \cite{BarFFS03,DemGKKPV06,DesFP03,DoFPS03,FlKKSS04,KlMP08,KrKM06,KrKM10,KrKSS03,TaZ14,YuY96}).
 The main obstacle for solving the problem
 is {\em symmetry}, which can occur
 at several levels
(topological structure, nodes,  agents, communication), each playing a key role in the 
difficulty of the problem and of its resolution.
For example, when the  network nodes are uniquely numbered, 
solving the gathering problem is trivial.  
On the other hand,
when the network nodes are anonymous, 
 the network topology is  highly symmetric, the mobile agents are identical,  
 and there is no means of communication, the problem is clearly 
 impossible
to solve by deterministic means.
The quest has been for  minimal  empowering assuptions which would make the problems
deterministically solvable.

A very common assumption is for the agents to have distinct {\em identities}
(e.g.,  see \cite{CzLP12,DemGKKPV06,DesFP03,YuY96}).
  This  enables
 different agents to execute  different deterministic algorithms; under such an
 assumption, the problem becomes solvable, and the focus is on the
 complexity of the solution.

An alternative type of assumption consists in
 empowering the agents with some minimal form of {\em explicit communication}.
 In one approach, this is achieved by having a whiteboard at each node
 giving the agents
 the ability to leave notes in each node they
travel (e.g., \cite{BarFFS03,ChDS07,DoFPS03}); in this case,
 some form of gathering can occur even in presence of some faults
\cite{ChDS07,DoFPS03}.   A less explicit and more primitive form of
communication is by 
endowing each agent with a constant number of  movable tokens, i.e. pebbles  that can be
placed on nodes, picked up, and carried while moving (e.g., \cite{CzDKK08}).

The less demanding assumption is that of  having 
 the homebases (i.e., the nodes where the agents are initially located)
identifiable by a mark, identical for all homebases, and
visible to any agent passing by it. This 
  setting is clearly much less demanding that  agents having identities  or
explicit communication;  originally suggested in
\cite{BasG91}, it has been used and studied
e.g., in  \cite{FlKKSS04,KrKSS03,Sawchuk04}.

Summarizing, the existing literature on the gathering and rendezvous problems
is extensive and the variety of
assumptions and results is aboundant (for  surveys see \cite{KrKM10,Pe12}).
However, regardless of their differences, all these investigations  
share the same fundamental assumption that the  topological structure does not change during the rendezvous or the gathering;
this is true also for those investigations   that  consider faulty nodes (e.g., see \cite{BoDD16,ChDS07,DoFPS03}).
In other words, they only consider {\em static graphs}.



Recently, within distributed computing,  researchers     started to investigate 
  {\em  dynamic graphs}, that is graphs where the topological changes are not
 localized and sporadic; on the contrary,  the topology  changes continuously and 
 at unpredictable locations, and  these changes are not anomalies (e.g., faults) but rather
  integral part of the nature of the system \cite{CaFQS12,KuO11}.

%
%

The study of distributed computations in  
highly dynamic graphs 
has concentrated on  problems of information diffusion, 
 reachability,  agreement, and exploration (e.g., \cite{Marge,IlKW14,IlW13,BiRSSW15,CaFMS14,HaK12,KuLO10,KuMO11}). 

In this paper we start the investigation of  gathering in dynamic graphs
by studying the feasibility of
this problem  in {\em  dynamic rings}.
Note that  rendezvous and gathering in a ring, the prototypical   symmetric graph,  have been
intesively studied in the static case (e.g., see the  monograph on the subject \cite{KrKM10}).
The presence, in the static case, of a mobile faulty agent that can block other agents, considered in  
\cite{DasLMMS16,DasLM15}, could be seen as inducing a particular form of dynamics.
Other than that,  nothing is known on gathering in dynamic rings.

\subsection{Main Contributions}

In this paper,  we study gathering of $k$  agents, identical and without communication capabilities,  in
 a dynamic ring
of $n$ anonymous nodes with identically marked homebases. The class of dynamics we consider is the
classic  {1-interval-connectivity} (e.g.,  \cite{DilDFS16,HaK12,KuLO10,KuMO11}); that is, the system is fully synchronous and  
under a (possibly unfair) adversarial schedule that,
at each time unit,  chooses which edge (if any) will be missing. Notice that this setting is not reducible to the one considered in 
 \cite{DasLMMS16,DasLM15}.

In this setting, we investigate under what conditions the gathering problem is
solvable. 
In particular, we focus on  
 the impact that   factors such as    {\em chirality} (i.e.,  common sense of orientation) and
  {\em cross detection} (i.e., the ability to detect, when traversing an edge, whether some agent is traversing it in the other direction), 
   have on the solvability of the problem.
Since, as we prove, gathering at a single node cannot be guaranteed  in a dynamic ring, 
we allow gathering to occur either at the same node, or at the two end nodes of the same link.


A main result of our investigation is the
 complete characterization of the classes  ${\cal F}(X,Y)$  of initial
configurations from which the gathering problem is solvable
with respect to chirality ($X\in\{\tt{chirality}, \neg\tt{chirality}\}$) 
and cross detection ($Y\in\{\tt{detection}, \neg\tt{detection}\}$).

In obtaining this characterization, we  establish several interesting results. For example, we show that,
without chirality,  cross detection is a
 powerful computational element; in fact, we prove (Theorems \ref{NoChirCross} and \ref{th:nocrossimposs}):
$${\cal F}(\neg\tt{chirality},\neg\tt{detection}) \subsetneq  {\cal F}(\neg\tt{chirality},\tt{detection})$$
Furthermore, in such systems 
 knowledge of the ring size $n$ cannot be substituted by knowledge of the number of agents $k$ (at least one of $n$ and $k$ must be known
 for gathering to be possible); 
in fact, we prove that  with cross detection but without chirality,  knowledge of $n$ is strictly more powerful 
than knowledge of $k$.


On the other hand, 
we show that,  with chirality, knowledge of $n$ can be substituted by knowledge of $k$,
yielding the same classes of feasible initial configurations.
Furthermore, with chirality, cross detection
is no longer a computational separator; in fact (Theorems \ref{lbl:th1} and \ref{lbl:th2})
$${\cal F}(\tt{chirality},\neg\tt{detection}) =  {\cal F}(\tt{chirality},\tt{detection})$$

We also observe that 
$${\cal F}_{static} = {\cal F}(\tt{chirality}, *) =  {\cal F}(\neg\tt{chirality},\tt{detection})$$
where ${\cal F}_{static}$ denotes   the set  of initial
configurations from which gathering is possible in the static case. In other words:
{\em with chirality or with cross detection, it is possible to overcome the computational
obstacles created  by the highly dynamic nature of the system}.

All the feasibility results of this characterization are  constructive: for each situation, we provide a distributed algorithm that allows
the agents to gather  within low polynomial time. 
In particular, the protocols for gathering with cross detection,
terminating in $O(n)$ time, are time {\em optimal}.
Moreover,  our algorithms are {\em effective}; that is, starting from any arbitrary configuration 
$C$ in a ring  conditions $X$ and $Y$, within finite time the agents   determine whether or not  $C\in {\cal F}(X,Y)$ is feasible,
and  gather if it is. See Figure \ref{table} for a summary of some of the results and the sections where they are established.

\begin{figure}[H]
\center
\begin{tabular}{ l | c | c |}
\cline{2-3}
              & \tt{no chirality} & \tt{chirality}  \\
\hline
\multicolumn{1}{ |c|  }{ \multirow{3}{*}{\tt{cross detection}} }& feasible:  ${\cal C} \setminus {\cal P}$ & feasible: ${\cal C} \setminus {\cal P}$   \\
  \multicolumn{1}{ |c|  }{\phantom{A}}      & time:  ${\cal O}(n)$ & time:  ${\cal O}(n)$   \\
  \multicolumn{1}{ |c|  }{\phantom{A}}              & (Sec. \ref{CrossNoChi}) & (Sec. \ref{CrossChi})  \\
        \hline
\multicolumn{1}{ |c|  }{ \multirow{3}{*}{\tt{no cross detection}    } }         & feasible: ${\cal C} \setminus ({\cal P} \cup {\cal E})$ &  feasible: ${\cal C} \setminus {\cal P}$  \\
    \multicolumn{1}{ |c|  }{\phantom{A}}     &time:   ${\cal O}(n^2)$ &time:   ${\cal O}(n \log n)$   \\
           \multicolumn{1}{ |c|  }{\phantom{A}}              &  (Sec. \ref{NoCrossNoChi}) &  (Sec. \ref{NoCrossChi})   \\
                     \hline
\end{tabular}
\caption{\small Each entry $(X,Y)$ shows the set ${\cal F}(X,Y)$ of feasible configurations,  and the time complexity of the gathering algorithm, where:
$X\in\{\tt{chirality},\neg\tt{chirality}\}$; $Y\in\{\tt{detection},\neg\tt{detection}\}$; and ${\cal C}$, ${\cal P}$ and  ${\cal E}$ are the set 
 of all possible  configurations, of the periodic  configurations, and of the configurations with an unique symmetry axis passing 
 through edges of the ring, respectively.}  \label{table}
\end{figure}

%
%
%
%
%


\section{Model and Basic Limitations}

\subsection{{\bf Model and Terminology}}
Let  ${\cal R} = (v_0,\ldots v_{n-1})$  be  
a synchronous dynamic ring where, 
at any time
step $t\in {\sc N}$,  one of its edges might not be present; the choice of which edge is missing (if any) is
controlled by an adversarial scheduler, not restricted by fairness assumptions.
Such a dynamic network is   known in the literature  as a {\em
1-interval connected}  ring.

Each node $v_i$  is connected to
its two neighbours $v_{i-1}$ and $v_{i+1}$ via distinctly  labeled ports
  $q_{i-}$ and   $q_{i+}$, respectively (all operations on
the indices are modulo $n$); the labeling of the ports is arbitrary and thus might not provide a
globally consistent  orientation. 
Each port of $v_i$ has an {\em incoming} buffer and an {\em outgoing} buffer. 
Finally, the  nodes are {\em anonymous} (i.e., have no distinguished
identifiers).  

\paragraph{Agents.}
Operating in  ${\cal R}$ is a  set ${\cal A}=\{a_{0},\ldots,a_{k-1}\}$  of
computational entities, called agents, each provided with  memory and computational capabilities.
The agents are {\em anonymous} 
(i.e., without distinguishing identifiers)
and all execute the same protocol. 

When in a node $v$, an agent can be {\em at} $v$ or  in one of the  port buffers. 
Any number of agents can
be in  a node at the same time;  an agent can determine how many other agents 
are in its location and where (in incoming buffer, in outgoing buffer, at the node).
Initially the agents are located  at   distinct   locations, called homebases; nodes that are homebases are
specially marked so that each agent can determine
 whether or not the current node is  a homebase.
 Note that, as discussed later, this assumption is necessary in our setting.

Each agent $a_j$  has a consistent private orientation $\lambda_j$ of the ring
which designates each port either {\em left} or {\em right}, with
$\lambda_j(q_{i-})=\lambda_j(q_{k-})$, for all  $0\leq i,k<n$.
The orientation of the agents might not be the same.  
If all agents agree on the orientation,  
we say that there is {\em chirality}.

The agents are {\em silent}: they  not have any explicit
communication mechanism.

The agents are {\em mobile}, that is they can move from node to
neighboring node.
More than one agent  may move on the same edge in the same direction in the same round. 

We say that the system has {\em cross detection} if whenever  two or more agents move in opposite directions on the same edge in the same 
round, the involved agents detect this event;  however they do not necessarily know the number of the involved agents in either direction.


\paragraph{Synchrony and Behavior.}
The system operates in {\em  synchronous} time steps, called {\em
rounds}.  
In each round, every agent is in one of a finite set of system states ${\cal S}$
which includes two special states: the initial state {\sf Init}  and the terminal state {\sf Term}.


At the beginning of a round $r$,  an  agent $a$ in  $v$ executes its  protocol (the same for all
agents).
Based on the number of agents at $v$ and in its buffers, and on the content of its
local memory and its state, it  determines whether or not to move  and, if so, in which direction
($direction\in \{left,right,nil\}$).

If $direction=nil$, the agent places itself at $v$ (if currently on a port).
If $direction \neq nil$, the agent moves  in the  outgoing buffer of the corresponding  port (if not already there);
 if the link is present, it arrives in the incoming buffer 
of the corresponding port of the destination node in round $r+1$; otherwise 
the agent does not leave the 
outgoing buffer.
As a consequence,  an agent can be in an outgoing
 buffer at the beginning of a round only when the corresponding link is not present.
 
In the following, when an agents is in an outgoing buffer that leads to the missing edge, we will say that the agent is {\em blocked}.



%

When multiple agents are at the same node,  all of them have 
the same direction of movement,  and are in the same state, we say that they form a {\em group}. 

 \paragraph{Problem Definition.}
Let  $({\cal R}, {\cal A})$ denote a system so defined. In  $({\cal R}, {\cal A})$,  gathering  is achieved in round $r$  if all agents in $A$  are
on the same node or on two neighbouring nodes in $r$;
in the first case,  gathering is said to be  {\em strict}.

An algorithm {\em solves}  \Gathering if, starting from any  configuration from which gathering is possible, 
within finite time all agents are in the terminal state,  are gathered, and  are aware that gathering has been achieved.

A solution algorithm  is {\em effective}  if  starting from any  configuration from which gathering is  {\em not} possible, 
within finite time all agents  detect such  impossibility.

\subsection{Configurations and Elections} 
The locations of the $k$ home bases in the ring is called a {\em configuration}.
Let  ${\cal C}$ be the set of all possible configurations with $k$ agents. 
Let $h_0, \ldots, h_{k-1}$ denote the nodes corresponding to the marked homebases (in a clockwise order) in  $C\in{\cal C}$.
We shall indicate by $d_i $ ($ 0 \leq i \leq k-1$) the distance (i.e., number of edges)  
between $h_i$ and $h_{i+1}$ (all operations are modulo $k$).
Let $\delta^{+j }$ denote the {\em inter-distance} sequence clockwise 
$\delta^{+j } = <d_j,d_{j+1} \ \ldots d_{j+k-1}>$, and let $\delta^{-j}$ denote
the couter-clockwise  sequence $\delta^{-j} = <d_{j-1}\ \ldots\ d_{j-(k-1)}>$.  
The unordered
pair of inter-distance sequences $\delta^{+j}$ and $\delta^{-j}$ describes   the    configuration
from the point of view of  node $h_j$.  

A  configuration is {\em periodic} with period $p$ (with $p | k$) if $\delta_i = \delta_{i+p}$ for all $i=0, \ldots k-1$.
Let ${\cal P}$ denote the set of periodic configurations.

Let $\Delta^+=\{\delta^{+j}:0 \leq j <k-1 \}$ and $\Delta^- =\{ \delta^{-j}: 0 \leq j <k-1\}$.
We will denote by $\delta_{min}$ the ascending lexicographically minimum
sequence in $\Delta^+\cup \Delta^-$. 
Among the non-periodic configurations,  particular ones are the  {\em double-palindrome} configurations, where   $\delta_{min} = \delta^{+i}  =  \delta^{-j} $ with $i\neq j$, where
it is easy to see that the two  sequences between the corresponding home bases  $h_i$ and $h_j$
are both palindrome.  A double-palindrome configuration has thus a unique axis of symmetry, equidistant from $h_i$ and $h_j$.
If such an axis passes through two edges (i.e., the distances between $h_i$ and $h_j$  are both odd),  we say that the configuration is {\em edge-edge}, and we denote by  ${\cal E}$   the set of edge-edge configurations. 
 
For example, let $k=4$ and $h_0, h_1, h_2, h_3$ be the four home bases with $d_0=3$, $d_1=4$, $d_2=5$, $d_3=4$. 
In this case we have   $\delta_{min} = \delta^{+0 } = \delta^{-1} = <3,4,5,4>$ and the unique axis of symmetry passes through two edges (one half-way between $h_0$ and $h_1$, the other half-way between $h_2$ and $h_3$).

A characterization of the configurations where a leader can be elected depending on chirality   is  well known in static rings.

\begin{property}
\label{election}
In a static ring without chirality, a leader node can be elected  from configuration $C$  if and only if   $C \in {\cal C}  \setminus   ({\cal P} \cup {\cal E})$; a leader edge can be elected if   and only if $C \in  {\cal C}  \setminus  {\cal P}  $.  \\
With chirality, a leader node can be elected   if and only if     $C \in {\cal C}  \setminus  {\cal P}  $.
\end{property}

Consider a ring without chirality. The {\em canonical} way to elect a leader from
 configuration   $C \in {\cal C}  \setminus   ({\cal P} \cup {\cal E})$  is described below.
If $C$  is asymmetric, the leader is the unique homebase that starts the lexicographically smallest inter-distance sequence.
If $C$ is double-palindrome, let $h$ and $h'$ be     the
 two homebases that start (in opposite direction) the two identical lexicographically smallest sequences:
if $C \in {\cal E}$  the leader edge  is the edge in the middle of the shortest portion of the ring delimited by  $h$ and $h'$ (note that both portions have odd distance and there is a central edge); otherwise  ($C \notin {\cal E}$) at least one of the two portions of the ring 
between $h$ and $h'$  has even distance and a central node is identified as the leader.

\subsection{{\bf Basic Limitations and Properties} \label{syncimp}}

Observe  that, in our setting, it is necessary for the homebases to be distingushable from the other nodes.

 \begin{property}
 If the homebases are not   distinguishable from the other nodes, then \Gathering  is unsolvable  in  $({\cal R}, {\cal A})$; this holds  regardless of chirality, cross detection,  and knowledge of $k$ and $n$.
\end{property}
 \begin{proof}
Let the  homebases be  not distinguishable from the other nodes  in   $({\cal R}, {\cal A})$.  To prove the property, it is sufficient to consider an execution in which all the entities have
the same chirality and no link ever disappears. Because of anonymity of the nodes and of the agents, and since homebases are not marked as such, in each round 
all agents will perform exactly the same action (i.e., stay still or move in the same direction). Thus the distance between  neighbouring agents will never change,
and hence  gathering will never take place if $k>2$. For $k=2$, by choosing the initial distance between the two agents to be grater than one,
 the same argument leads to the same result.
 \end{proof}

\noindent Thus, in the following we assume that the homebases are identical but distinguishable from the other nodes. 

An obvious, very basic limitation that holds even if the ring is static is the following.

 \begin{property}\label{nORk}
 In  $({\cal R}, {\cal A})$, if neither $n$ nor $k$  are known, then \Gathering   is  unsolvable; this holds  regardless of chirality and cross detection.
\end{property}

\noindent Hence at least one of $n$ or $k$ must be known.

An important limitation follows from the dynamic nature of the system:
 
 \begin{property}\label{obs-noNodeGathering}
 In  $({\cal R}, {\cal A})$, strict \Gathering   is unsolvable; this holds  regardless of chirality, cross detection,  and knowledge of $k$ and $n$.
\end{property}
\begin{proof}
Consider the following strategy of the adversarial scheduler. 
It selects two arbitrary agents, $a$ and $b$; at each round,  
the adversary will not remove any edge, unless the two agents would meet in the next step.  More precisely,
if the two agents would meet by both independently moving on different edges $e'$ and $e''$ leading to the same vertex,
then the adversary removes one of the two edges;
if instead  one agent is not moving from its current node $v$, while the other is moving on an edge $e$ incident to $v$, the adversary removes edge $e$.
This strategy ensures  that $a$ and $b$ will never be at the same node at the same time.
\end{proof}

\noindent Hence, in the following we will not require gathering to be strict. 

An obvious but important  limitation, inherent to the nature of the problem, holds even in static situations:

\begin{property} 
\label{th:noPeriodic}
\Gathering is unsolvable if the initial configuration $C\in {\cal P}$; this holds  regardless of chirality, cross detection,  and knowledge of $k$ and $n$.
\end{property}
 \begin{proof}
It is sufficient to consider an execution in which all the entities have the same chirality and no link ever disappears. Depending on their initial positions,
the agents can be partitioned  into $k/p$ congruent classes, where $p$ is the period of the initial configuration, each composed of $p$ agents.
In each round,
all agents of the same class  will perform exactly the same action (i.e., stay still or move in the same direction) based on the same observation.
 Thus the distance between   two consecutive agents of the same class will never change;
 hence  gathering will never take place.
 \end{proof}

\noindent Hence, in the following we will focus on initial configurations not in ${\cal P}$.


\section{General Solution Stucture}

The  solution algorithms for gathering have the same general structure, and use the same building block and variables.

\paragraph{General Structure.}
All the algorithms are divided in two phases.  The goal of {\sf Phase 1} is for  the agents to explore the ring. In doing so, they may happen to solve \Gathering as well.
If they complete   Phase 1 without gathering, the agents are able to elect a node or an edge (depending on the specific situation) and  the algorithm proceeds to {\sf Phase 2}. In Phase 2 the agents try to gather around the elected node (or edge); however, gathering  on that node (or edge) might not be possible due to the fact that the agents cannot count on the presence of all edges at all times. Different strategies are devised, depending on the setting, to guarantee that in finite time the problem is solved in spite of the choice of schedule of missing links decided by the adversary.
For each setting, we will describe the two phases depending on the availability or lack of   cross detection, as well as on the presence or not of  chirality.
Intuitively, cross detection is useful to simplify termination in Phase 2,  chirality helps in breaking symmetries.

\paragraph{Exploration Building Block.}
At each round, an agent evaluates a set of predicates: depending on the result of this evaluation, it chooses a direction of movement and possibly a new state. In its most general form, the  evaluation of the predicates occurs through the building block procedure \Explore($dir$ $|$ $p_1:s_1$; $p_2:s_2$; \dots; $p_h:s_h$),
where $dir$ is either \dLeft~or \dRight,  $p_i$ is a predicate,  and $s_i$ is a state. 
In Procedure  \Explore,  the agent evaluates the predicates $p_1, \ldots, p_h$ in order;
as soon as a predicate is satisfied, say $p_i$, the procedure exits and the agent does a transition to the specified state, say $s_i$. If no predicate is satisfied, the agent tries to move in the specified direction $dir$ and the procedure is executed again in the next round. 
In particular, the following predicates are used:
\begin{itemize}
\item \pMeeting, satisfied when the agent (either in a port or at a node) detects an increase in the numbers of agents it sees at each round.
\item \pSameDirMeeting, satisfied when the agent detects, in the current round, new agents moving in its same direction. This is done by seeing new agents in an incoming or outgoing buffer corresponding to a direction that is equal to the current direction of the agent. 

\item \pOppositeDirMeeting, satisfied when the agent detects, in the current round, new agents moving in its opposite direction. This is done by seeing new agents in an incoming or outgoing buffer corresponding to a direction that is opposite to the current direction of the agent. 
\item \pCross, satisfied when the agent, while traversing a link, detects in the current round other agent(s) moving on the same link in the opposite direction.

\item $\emph{ seeElected}$: let us assume there is either an elected node or an elected edge. This predicate is satisfied when the agent has reached the elected node or an endpoint of the elected edge.
\end{itemize}

Furthermore, the agents keeps six variables during the execution of the algorithm. Two of them are never reset during the execution; namely:
\begin{itemize}
\item \vTtime: the total number of rounds since the beginning of the execution of the algorithm (initially set to $0$);
\item \vTotOthers: the number of total agents (initially set to $0$). This variable will be set only after the agent completes a whole loop of the ring, and will be equal to $k$.
\end{itemize}

Other four variables are periodically reset; in particular: 

\begin{itemize}
\item $r_{ms}$: it stores the last round when the agent meets someone (at a node) that is moving in the same direction (initially set to $0$); this value is updated each time a new agent is met, and it is reset at each change of state or direction of movement;  
\item \vBtime: the number of rounds the executing agent has been blocked trying to traverse a missing edge since $r_{ms}$. This variable is reset to $0$ each time the agent either
traverses an edge or changes direction to traverse a new edge;
\item \vEtime, \vEsteps: the total number of rounds and edge traversals, respectively. These values are reset at each new call of procedure \Explore or when $r_{ms}$ is set.
\item \vOthers: the number of agents at the node of the executing agent. This value is set at each round.
\end{itemize}

\section{Gathering With Cross Detection}

In this section,  we study gathering  in dynamic rings   when  there is cross detection; that is,  
 an agent crossing a link  can detect whether other agents are crossing it  in the opposite direction. 
 Recall that, by Property \ref{nORk}, at least one of $n$ and $k$ must be known.
 
We  first examine the problem  without chirality and show that, with knowledge of $n$,
it is sovable  in all  configurations that are feasible in the static case; furthermore, this is done in
optimal time $\Theta(n)$. On the other hand, with  knowledge of $k$ alone, 
the problem is unsolvable.

We  then examine the problem with chirality, and show that in this case the problem
 is sovable  in all  configurations that are feasible in the static case even with   knowledge of $k$ alone;
furthermore, this is done in
optimal time $\Theta(n)$.

\subsection{With Cross Detection: Without Chirality}\label{CrossNoChi}

In this section, we present and analyze the algorithm, {\sc Gather(Cross,$\not$Chir)}, that 
solves \Gathering  in rings of known size with cross setection but without chirality.


The  two phases of the algorithm are described and analyzed below.\\

\subsubsection{Algorithm {\sc Gather(Cross,$\not$Chir)}: {\sf Phase 1}}\label{ph1:knownnochir}

The overall idea of this phase, shown in Figure~\ref{algp1:gatheringnochircross}, is to let the agents move long enough along the ring to guarantee that, if they do not gather, they all manage to fully traverse the ring in spite of the link removals.

 More precisely,  for the first $6n$ rounds each agent  attempts to move  to the left (according to its orientation). 
 At round $6n$, the agent checks if the predicate \predC$\equiv (r_{ms} < 3n \land \vEsteps < n)$ is verified. 
If \predC is not verified, then (as we show)  the agent has explored the entire ring and thus
knows  the total number $k$ of  agents  (local variable \vTotOthers);
 in this case, the agent switches direction, and enters state {\sf SwitchDir}.
 Otherwise, if after $6n$ rounds \predC is satisfied, then $k$ is not known yet: in this case, the agent keeps the same direction, 
and enters state {\sf KeepDir}.

\begin{figure}[H]
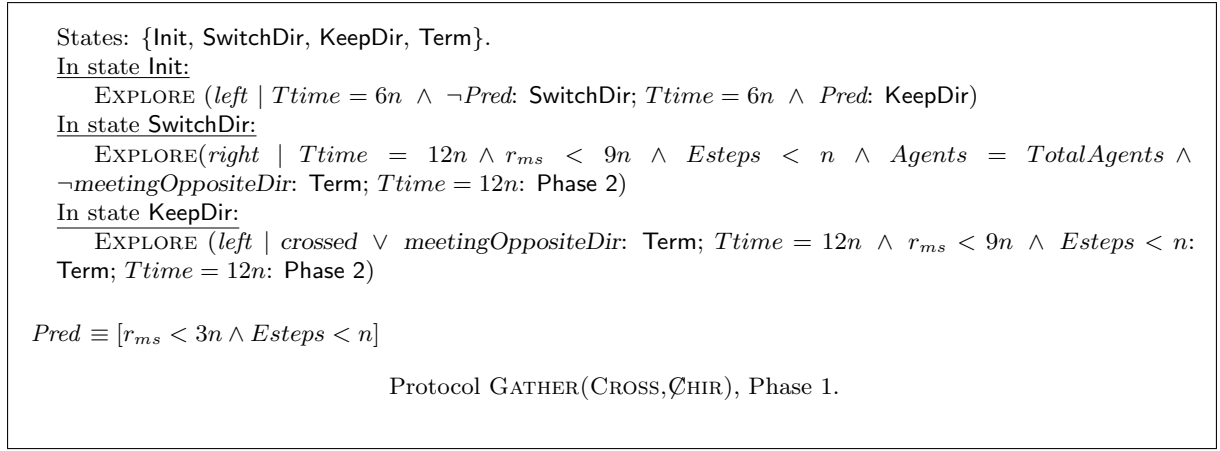

\footnotesize
\begin{framed}
\begin{algorithmic}
\State States: \{{\sf Init}, {\sf SwitchDir}, {\sf KeepDir}, {\sf Term}\}.
\AtState{Init}
    \State \Explore({\dLeft$|$ \vTtime$=6n \,\, \land \,\,  \neg \predC$: {\sf SwitchDir}; $\vTtime=6n \,\, \land \,\, \predC$: {\sf KeepDir}})
\AtState{SwitchDir}
    \State \Call{Explore}{\dRight $|$ $ \vTtime = 12n \land r_{ms} < 9n \,\, \land \,\, \vEsteps < n \,\, \land \,\, \vOthers = \vTotOthers \land \neg \pOppositeDirMeeting$:  {\sf Term};  $\vTtime = 12n$: {\sf Phase 2}}
\AtState{KeepDir}
    \State \Explore({\dLeft$|$ $\pCross \,\, \lor \,\, \pOppositeDirMeeting$: {\sf Term}; $\vTtime = 12n \,\, \land \,\, r_{ms} < 9n \,\, \land \,\, \vEsteps < n$:  {\sf Term};  $\vTtime = 12n$: {\sf Phase 2}})
    
  \END
\end{algorithmic}
 \predC$\equiv [r_{ms} < 3n \land \vEsteps < n]$
\begin{center}Protocol  {\sc Gather(Cross,$\not$Chir)}, Phase 1.  \end{center}

\end{framed}
\caption{{\sf Phase 1} of  Algorithm 
{\sc Gather(Cross,$\not$Chir)}  
\label{algp1:gatheringnochircross}}
\end{figure}


In state {\sf SwitchDir}, the agent
attempts to move  in the chosen direction until round $12n$. 
 At round $12n$, the agent terminates if the predicate $[r_{ms} < 9n \land \vEsteps< n]$ holds,   predicate \pOppositeDirMeeting does not hold, and
  in its current node there are $k$ agents; otherwise, it starts {\sf Phase 2}.

In state  {\sf KeepDir},  if at any round predicate $\pCross$ or predicate $\pOppositeDirMeeting$ hold,
the agent   terminates; otherwise, it attempts to move to its left until round $12n$. At round $12n$, if the predicate $[r_{ms} < 9n \land \vEsteps< n]$ holds, 
the agent
 terminates; otherwise, it switches to {\sf Phase 2}.

We now prove some important properties of {\sf Phase 1}. 

\begin{lemma}\label{observationmovements}
Let agent $a^*$  move less than $n$ steps in the first $3n$ rounds.
Then, by round $3n$, all agents moving in the same direction as 
$a^*$ belong to the same group.
\end{lemma}
\begin{proof}
Let us focus only on the set $A$ of the agents that have the same orientation of the ring as $a^*$.
In the first $6n$ rounds of Phase 1, each agents attempts to move in the same direction.
if there is a round $r\leq 6n$ when $a^*$ is blocked, then every $a\in A$
that at round $r$ is not at the same node of $a^*$ does move, due to the 1-interval connectivity of the ring.
Since  $a^*$ moves less than $n$ steps in the first $3n$ rounds,
then the number of rounds in which $a^*$ is blocked is greater than $2n+1$. 
Thus, all agents in $A$ that are not already
in the same node as $a^*$ 
have moved towards $a^*$ of $2n+1$ steps. 
On the other hand, everytime  $a^*$ moves, the other agents might be blocked; however,  by hypothesis, this has happend 
 less than $n$ times.
 
 Since  the initial distance between $a^*$ and  an agent  in $A$  is at most $n-1$, it follows
 such a distance   increases less than $n$ (due to $a^*$ moving), but it  decreases by $2n+1$ (due to $a^*$ being blocked);
 thus the distance is  zero (i.e., they are at the same node)  by round $3n$.
\end{proof}

Because of absence of chirality, the set ${\cal A}$ of agents can be partitioned into two sets where all the agents in the same set 
share the same orientation of
the ring; let $A_r$ and $A_l$ be the two sets.

\begin{lemma}\label{lemmaC}
Let  $A\in\{A_r,A_l\}$. If at round $6n$ \predC is verified for an agent $a^* \in A$, then all agents in $A$ are in the same group at round $6n$. 
Moreover, \predC is verified for all agents in $A$. 
\end{lemma}
\begin{proof}
By definition of \predC and by Lemma~\ref{observationmovements}, at round $6n$ all agents in $A$ are at the same node of $a^*$. 
Also, let $r$ be the first round when all agents in $A$ meet at the same node: by definition, the value of $r_{ms}$ for all agents under consideration is exactly $r$. 
From this observation and since \predC holds for $a^*$, it follows that \predC must be satisfied for all agents in $A$. 
\end{proof}

\begin{lemma}\label{lemmaAB}
Let  $A\in\{A_r,A_l\}$. If \predC is not verified at round $6n$ for  agent $a^* \in A$, then  at round $6n$ all agents in $A$ 
have done a complete tour of the ring (and hence know the number of total agents, $k$); moreover,  \predC is not verified for all agents in $A$. 
\end{lemma}
\begin{proof}
Let us assume by contradiction that there exists $a' \in A$ that has not done a complete tour of the ring after $6n$ rounds; that is, $a'$ has moved less then $n$ steps in the first $3n$ rounds. By Lemma~\ref{observationmovements}, all agents in $A$ are in the same node as $a^*$  by round $r < 3n$. Therefore, \predC would be satisfied for any of the robots in $A$, including $a^*$:  a contradiction.  

To prove the second part of the lemma, note that \predC cannot be satisfied for any agent in $\in A$: in fact, by Lemma~\ref{lemmaC}, this would prevent the existence of an agent in $A$ for which \predC  is not satisfied. Thus, the lemma follows. 
\end{proof}

\begin{lemma}\label{phase1}
If one agent terminates in {\sf Phase 1}, then all agents terminate and gathering has been correctly achieved. Otherwise, no agent terminates and all of them have done a complete tour of the ring. 
\end{lemma}
\begin{proof}
Notice that, by construction, the agents do not change their direction before  round $6n$.

Let us first consider the case when at round $r=0$ the agents do not have the same orientation.
We distinguish three possible cases, depending on what happens ar round $6n$.
 
\begin{enumerate}
\item {\em At round $6n$, all agents change direction.} 
By Lemma~\ref{lemmaC}, it follows that at round $6n$ all of them completed a loop of the ring. 
According to {\sf SwitchDir}, an agent, to enter the {\sf Term} state, has to verify both (a) $\vOthers = \vTotOthers$ and (b) $\neg \pOppositeDirMeeting$: to verify (a), the agents have to meet at the same node, thus $\pOppositeDirMeeting$ has to be true, hence (b) can not verified. 
It follows that the agents cannot terminate at round $12n$, and the lemma follows. 

\item {\em At round $6n$, no agent changes direction.} Thus, according to the algorithm, \predC is verified for all agents, that will enter {\sf KeepDir} state; also, by Lemma \ref{lemmaC}, all agents that share the same direction are in the same group (i.e., there are two groups of agents moving in opposite direction).

By definition of {\sf KeepDir}, if between round $6n$ and $12n$ an agent crosses or meets another agent, they both terminate; hence, all the agents in their respective group terminate, and the lemma follows.
If no crossing occurs between round $6n$ and $12n$, then both group of agents are necessarily blocked at the ends of the missing link (otherwise the two groups would have crossed or met). Thus, at round $12n$, $r_{ms} < 9n$ (last reset of $r_{ms}$ occurred at round $6n$) and $\vEsteps < n$ (otherwise, again, the two groups would have crossed or met), for any agent; hence all agents terminate at round $12n$, and the lemma follows.

\item {\em At round $6n$, only some agents  change direction.}  By Lemmas \ref{lemmaC} and \ref{lemmaAB}, it follows that, after round $6n$. all agents will move in the same direction. 

Let us assume that, at round $12n$, condition $r_{ms} < 9n \land \vEsteps < n$ holds for some agent $a^*$. 
If $a^*$ did not switch direction at round $6n$, $a^*$ terminates at round $12n$, say at node $v$ ({\sf KeepDir}); 
hence, by Lemma~\ref{observationmovements}, all agents gather at $v$.
Otherwise, if $a^*$ switched direction at round $6n$, since all agents are moving in the same direction, condition $\pOppositeDirMeeting$ is false from round $6n$ on; moreover, by Lemma~\ref{lemmaAB}, $a^*$ computed the number $k$ of total agents at round $6n$. 
Therefore, $a^*$ terminates at round $12n$, say at node $v$ ({\sf SwitchDir}). Finally, by Lemma~\ref{observationmovements}, 
all agents gather at $v$, and the lemma follows.

On the other hand, if condition $r_{ms} < 9n \land \vEsteps < n$ does not hold for any agent at round $12n$, no agent can enter the {\sf Term} state. Also, following an argument similar to the one used in Lemma~\ref{lemmaAB}, we have that all agents have done a complete loop of the ring after $6n$ rounds, and the lemma follows.
\end{enumerate}

The  other case left to consider is  when at round $r=0$ the agents have the same orientation. We distinguish two cases.
\begin{enumerate}

\item {\em There is an agent that does not change direction at round $6n$.}  Then, at this time, all agents are in the same group 
and none of them switches direction (Lemma \ref{lemmaC}). Thus, if the agents terminate at round $12n$, gathering is solved, 
and the lemma follows. Otherwise, by {\sf KeepDir}, predicate $r_{ms} < 9n \land \vEsteps < n$ is not verified at round $12n$ 
for any of them (they are all in the same group) and they have all done a complete loop of the ring 
(last reset of $r_{ms}$ occurred at round $6n$, hence $\vEsteps \geq n$ for all agents), 
so they  start {\sf Phase 2}, and the lemma follows.

\item {\em There is an agent that switches direction at round $6n$.} Then, at this time, all agents are in the same group, all of them switch direction, and have done a complete loop of the ring (Lemma \ref{lemmaAB}). The proof follows with an argument similar to the one of previous case.
\end{enumerate}

\end{proof}

 \subsubsection{Algorithm {\sc Gather(Cross,$\not$Chir)}: Phase 2}\label{ph2:knownnochir}
If the agents execute {\sf Phase 2} then,  by Lemma~\ref{phase1}, they  know both the position of all the homebases and the number of agents $k$;
that is, they know the initial configuration $C$. 
 If $C\in {\cal P}$,  gathering is impossible (Property \ref{th:noPeriodic}) and they become aware of this fact. Otherwise,
if $C\in {\cal E}$ they can elect an edge $e_L$,  and if   $C\in {\cal C}\setminus ({\cal P}\cup {\cal E})$ they can elect one of the node as leader $v_L$ 
(Property \ref{election}).
For simplicity of exposition and without loss of generality, in the following we assume that {\sf Phase 2} of the algorithm, shown in Figure \ref{algp2:gatheringnochircross}, starts at round $0$.

\begin{figure}[H]
\footnotesize
\begin{framed}
\begin{algorithmic}
\State States: \{{\sf Phase 2}, \sReached,  \sReaching,  {\sf Joining}, {\sf Waiting},  {\sf ReverseDir},{\sf Term}\}.
\AtState{Phase 2}
 \If {$C \in {\cal P}$}
 \State {\sf unsolvable}()
    \State {\sf Go to State} {\sf Term}
 \EndIf
  \State resetAllVariables except \vTotOthers
   \State $dir=${\sf shortestPathDirectionElected}()
    \State \Explore({$dir\,\,|$ $\emph{ seeElected} $: \sReached; $\vTtime=3n$: \sReaching  })
    \AtState{\sReached}
    \State $dir=${\sf opposite}$($dir$)$
    \If {\vTtime$\geq 3n$}
    \State \Explore({$dir \,\,|$ $ \vOthers =\vTotOthers \lor \vBtime = 2n$:  {\sf Term};  $ \pCross $:  {\sf Joining}})
    \EndIf
\AtState{Joining}
 \State $dir=${\sf opposite}$($dir$)$
    \State \Explore({$dir \,\,|$ $\vOthers =\vTotOthers \lor \vBtime = 2n \lor \pCross$: {\sf Term};    $\vEsteps =1$: {\sf ReverseDir}})
\AtState{\sReaching}
    \State \Explore({$dir \,\,|$ $\vOthers =\vTotOthers \lor \vBtime = 2n$: {\sf Term}; $\pSameDirMeeting$: \sReached;  $\pOppositeDirMeeting\ \lor$ \emph{seeElected}: {\sf ReverseDir}; $\pCross$: {\sf Waiting}})
   \AtState{Waiting}
    \State \Explore({nil $|$ $ \vEtime > 2n$:  {\sf Term};  $\pMeeting$: {\sf ReverseDir}})
       \AtState{ReverseDir}
       \State $dir=${\sf opposite}$($dir$)$
       \State {\sf Go to State} \sReached
  \END
\end{algorithmic}
\begin{center}Protocol  {\sc Gather(Cross,$\not$Chir)}, Phase 2.  \end{center}
\end{framed}
\caption{{\sf Phase 2} of  Algorithm 
{\sc Gather(Cross,$\not$Chir)}  
\label{algp2:gatheringnochircross}}
\end{figure}

In {\sf Phase 2}, an agent first resets all its local variables, with the exception of \vTotOthers, that stores the number of agents $k$; between rounds $0$ and $3n$, each agent moves toward the elected edge/node following the shortest path ({\sf shortestPathDirectionElected}()). 
If at round $3n$ an agent has reached the elected node or an endpoint of the elected edge it stops, and enters the \sReached state.
Otherwise (i.e., at round $3n$, the agent is  not in state \sReached), it switches to the \sReaching state. 
If all agents are in the same state (either $\sReached$ or $\sReaching$), then they are in the same group, and terminate ($\vOthers=\vTotOthers)$.
If they do not terminate, all agents start moving: each  $\sReaching$ agent in the same direction it chose at the beginning of {\sf Phase 2},  while the
 \sReached agents reverse direction.
 
  From this moment,
each  agent,   regardless of its state,  terminates immediately  if all  $k$ agents are in the same node,  
or if  it is blocked on a missing edge for $2n$ rounds.  In other situations, the behaviour of each agent $a^*$ depends on its state,  as follows.

\paragraph{State \sReached.} 
If $a^*$ crosses a group of agents, it enters the {\sf Joining} state. In this new state, say at node $v$, the agent 
switches direction in the attempt to catch and join the agent(s)  it just crossed. 
 If $a^*$  leaves $v$  without crossing any agent  ($\vEsteps =1$), 
 $a^*$ enters again the \sReached state, switching again direction (i.e., it goes back to direction originally chosen when {\sf Phase 2} started).
If instead $a^*$  leaves $v$ and it crosses  some agents, it terminates: this can happen 
because also the agents that $a^*$  crossed   try to catch  it  (and all other agents  in the same group with $a^*$). 
As we will show, in this case all agents can correctly terminate.

\paragraph{State \sReaching.} 
If $a^*$  is able to reach the elected node/edge (\emph{seeElected} is verified), it enters the \sReached state, and switches direction. 

If $a^*$ is blocked on a missing edge and it is reached by other agents, then it switches state to \sReached keeping its direction (\pSameDirMeeting is verified).

Finally, if $a^*$ crosses someone, it enters the {\sf Waiting} state, and it stops moving. 
 If while in the {\sf Waiting} state $a^*$ meets someone new before $2n$ rounds, it enters the \sReached state, and switches direction. 
Otherwise,  at round  $2n$ round it terminates. 
 
 \begin{lemma}\label{observation2}
 At  round $3n$ of {\sf Phase 2}, there is at most one group of agents in state \sReaching, and at most two groups of agents in state \sReached.
 \end{lemma}
 \begin{proof}
 In {\sf Phase 2}, all agents start moving towards the nearest elected endopoint/node. The lemma clearly follows for agents in state \sReached: in fact, the two
 groups (one of them possibly empty) are formed by all the agents that have successfully reached the elected endpoint/node from each of the two directions.

If an agent is not able to reach the elected endpoint/node within $3n$ rounds, it must have been blocked for at least $2n+1$ rounds; 
notice that this cannot happen to two agents walking on disjoint paths toward the elected endpoint/node. 
Therefore, by Lemma~\ref{observationmovements}, there can be at most one group of agents in state \sReaching, and the lemma follows. 
 \end{proof}
 
Note that, if at round $3n$ there are two groups of agents in state \sReaching, they have opposite  moving direction \emph{dir}; 
also, they are either at the same leader node, or at the two endpoints of the leader edge. 

 \begin{lemma}\label{lemmaphase2}
If an agent $a^*$ terminates executing {\sf Phase 2}, then all other agents will terminate, and gathering is correctly achieved. 
 \end{lemma}
 \begin{proof}
If $a^*$ terminates because $\vOthers =k$, the lemma clearly follows. Let us consider the other termination conditions.
\begin{enumerate}
\item {\em $a^*$ is either in state \sReached or \sReaching, and $\vBtime = 2n$.}
 Agent $a^*$ is blocked on one endpoint of the missing edge; thus, after $2n$ steps, all agents with opposite direction are on the other endpoint of the missing edge. Note that this holds also if the other agents are all in the \sReaching state and reach the elected endpoint/node in (at most) $n$ rounds: in this case, in fact, they would switch direction, and go back to the other endpoint of the missing edge in at most other $n$ steps. 

Therefore, the other agents will either terminate because they wait for $2n$ rounds at the other endpoint of the missing edge, or because they reach the same endpoint node where $a^*$ terminated ($\vOthers=\vTotOthers$ is thus verified); hence they correctly gather, and the lemma follows. 

\item {\em $a^*$ is in state {\sf Joining} and $\pCross$ is verified.}
 First notice that, if $a^*$ crosses some agent(s), then the crossed agent(s) are in state {\sf Joining} as well (the agents in the {\sf Waiting} state do not try to actively cross an edge); thus, they were in the \sReached state before crossing. However, this is possible only if there is no group of agents in state \sReaching: at round $3n$, the two groups \sReached starts moving in opposite directions from the same node or from two endpoints of the same edge. Therefore, when they cross, one of them has already met the group \sReaching, if it exists, and when that happens the group \sReaching merges with the group \sReached. This implies that, when two groups \sReached cross, all agents are in {\sf Joining}. Therefore, when they cross again, all agents are on the two endpoints of the same edge, and the lemma follows.

\item {\em $a^*$ is in state {\sf Waiting}, and $\vEtime > 2n$.} By Lemma~\ref{observation2}, $a^*$ has crossed a group of agent in state \sReached. These agents, by entering the {\sf Joining} state, actively try to reach the node where $a^*$ is (in {\sf Waiting}).  If the {\sf Joining} group does not reach $a^*$ in $2n$ rounds, then the edge connecting them is necessarily missing. Also note that, if there is another \sReached group, it has to reach the agents in the {\sf Joining} state within $2n$ rounds. Now, these two groups will either terminate by waiting $2n$ rounds, or because they are able to reach the {\sf Waiting} agent $a^*$, finally detecting that $\vOthers=\vTotOthers$. In all cases, the agents correctly terminate solving the gathering, and the lemma follows. 
\end{enumerate}
 \end{proof}

  \begin{lemma}\label{phase2rounds}
 {\sf Phase 2} terminates in at most $10n$ rounds. 
 \end{lemma}
 \begin{proof}
By Lemma~\ref{observation2}, at the end of round $3n$, the following holds:
 \begin{enumerate}
 \item If there is only one group with state \sReaching, the agents terminate on condition $\vOthers =\vTotOthers$.
 \item If there is only one group with state \sReached, the agents terminate on condition $\vOthers =\vTotOthers$.
 \item If there are two groups with state \sReached, they have opposite direction of movements (otherwise, they would be in the same group). Therefore, within $n$ rounds, they have to be at distance $1$ from each other: they terminate within the next $2n+1$ rounds either by crossing in state {\sf Joining}, or on condition $\vBtime = 2n$.
 \item If there are two groups of agents in the \sReached state, say $G$ and $G'$, and one group of agents in the \sReaching state, say $G^*$, then $G$ and $G'$ have opposite direction of movements (otherwise, they would be in the same group); hence one of them, say $G$, has direction of movement opposite to the one of $G^*$. Therefore, within $n$ rounds, $G$ and $G^*$ have to be at distance $1$ from each other. If they do not cross each other within the next $2n$ rounds, they will terminate on condition $\vBtime = 2n$, and the lemma follows. 
 
Otherwise (they cross within the next $2n$ rounds), two cases can occur: (A) they both terminate, one group on condition $\vBtime = 2n$ and the other one on condition $ \vEtime > 2n$ in the {\sf Waiting} state (between the two groups there is the missing edge); or (B) they will join within the next $2n$ rounds. In Case (A) the lemma follows. In Case (B), they either terminate on condition 
 $\vOthers =\vTotOthers$, and the lemma follows; or the \sReaching group enters the \sReached state (via {\sf Waiting}), and starts moving towards the other \sReached group. In this last case, the proof follows from previous Case 3.
 
 \item If there is one group in the \sReached state and one in the \sReaching state, we have two possible cases. (A) The two groups are moving towards each other: in this case the proof follows similarly to the previous Case 3. (B) The two groups move in the same direction. If the group \sReaching does not reach the elected endpoint/node within $2n+1$ rounds, the two groups necessarily meet, and thus terminate; hence the lemma follows. Otherwise, after \sReaching reaches the elected endpoint/node, this group enters the \sReached state, and the proof follows similarly to the previous Case 3.
 \end{enumerate}

 \end{proof}

Hence we have
 \begin{theorem} 
 \label{NoChirCross}
Without chirality, \Gathering is solvable in rings of known size with cross detection, starting from any $C\in {\cal C}\setminus{\cal P}$. 
Moreover, there exists an algorithm solving \Gathering that terminates in ${\mathcal O(n)}$ rounds for any $C\in {\cal C}\setminus{\cal P}$
and,  If $C\in {\cal P}$, the algorithm  detects that the configuration is periodic. 
 \end{theorem}
 \begin{proof}
 If  algorithm {\sc Gather(Cross,$\not$Chir)} terminates in {\sf Phase 1}  then, by Lemma~\ref{phase1}, it correctly solves gathering and it terminates by round $12n$.
 If it terminates in {\sf Phase 2}, then by Lemma \ref{lemmaphase2}, it correctly solves gathering, and by lemma \ref{phase2rounds}
 will do so in at most $10$ additional rounds.
  Notice, that in {\sf Phase 1}, either the agents discover the initial configuration $C$ or they gather. 
  Once they know $C$, they can detect if the problem is solvable or not. This proves the last statement of the theorem
  \end{proof}
 
\subsection{Knowledge of $n$ is more Powerful Than Knowledge of $k$}

One may ask if it is possible to obtain the same result of Theorem \ref{NoChirCross} if
knowledge of  $k$ was available instead of  $n$; recall that at least one of $n$ and $k$ must be known (Property  \ref{nORk}).
Intuitively, knowing $k$,  if an agent manages to travel all along the ring,  it will discover also the value of $n$.
Unfortunately, the following Theorem shows that, from a computational point of view,  knowledge of the ring size is strictly more powerful than
knowledge of the number of agents. 

\begin{theorem}\label{theorem:impth2}
In rings with no chirality, \Gathering  is impossible without knowledge of $n$  when starting 
from a configuration $C \in {\cal E}$.  
This holds even if there is cross detection and $k$ is known.
\end{theorem}
\begin{proof}
 By contradiction.
 Let us suppose to have two agents  $a$ and  $b$  on a ring $R$ where the the   distances between the homebases $h_1$ and $h_2$
 are   $d_1< d_2$  and they are     both odd. 
  Let $e_1$ be the central edge between $h_1$ and $h_2$  in the smallest portion of the ring (i.e., at distance  ${{(d_1 -1)} \over 2}$ from $h_1$ and $h_2$) and 
    $e_2$ the central edge on the other side  (i.e., at distance  ${{(d_2  -1)} \over 2}$ from $h_1$ and $h_2$).
 Let us consider an execution $E$ of a correct  algorithm ${\cal A}$ starting from this configuration.  
The adversary decides opposing clockwise orientation for the two agents, and it only removes edges $e_1$ and $e_2$ during the execution of the algorithm. We will show that, by appropriately removing only this two edges, the adversary can prevent the two agents to ever see each other.
 At the beginning the agents   moves towards each other (w.l.o.g, in the direction of $e_1$).
  The adversary lets them move until they are about to traverse   edge $e_1$; at this point   edge $e_1$ is removed and     both agents are blocked  with 
  symmetric histories. After a certain amount of time,  they will either  both reverse direction or terminate. The same removal scheduling is taken whenever 
   they are about to cross either $e_1$ or  $e_2$.
  The adversary keeps following this schedule until both agents decide to terminate. Notice that for $A$ to be correct they can only terminate on the endpoints of one of the  edges $e_1$ or $e_2$. Let $r'=f(R)$ be the round  when  the agents terminate in    execution $E$.

Let us now consider the same algorithm on  a ring $R'$ of size greater than $4f(R)+2$ where the two agents are initially  placed 
at distance greater than $2f(R)$.  
Consider  agent $a$:  the adversary  removes the  edge  at distance $\frac{d_{1}-1}{2}$ on its right 
 and the one at  distance  $\frac{d_2-1}{2}$ on its     left whenever $a$ tries to traverse them. In doing so  
 $a$ does not perceive any difference  with respect to execution $E$,  and 
therefore terminates at round $r'=f(R)$. 
At this point, the other agent $b$ cannot be at the other extreme of the edge where $a$ terminated, 
therefore,   the adversary now  blocks $b$ from any further move, preventing  gathering.
A contradiction.
\end{proof}

\subsection{With Cross Detection: With Chirality}\label{CrossChi}

Let us now consider the simplest setting,  where the agents have cross detection capability as well as   a common  chirality.
In this case, the impossibility result of the previous Section does not hold, 
and a solution to \Gathering exists  also when $k$ is known  but $n$ is not.

The solution consists of  a simplification of Phase 1 of   Algorithm {\sc Gather(Cross,$\not$Chir)}, also extended to the case of $k$ known, followed by Phase 2 of   Algorithm {\sc Gather(Cross,$\not$Chir)}.


\subsubsection{Algorithm {\sc Gather(Cross, Chir)}: {\sf Phase 1}}\label{ph1:knownchir}

In case of known $n$, each agent  executes  Phase $1$ of Algorithm {\sc Gather(Cross,$\not$Chir)}  moving clockwise
until round $6n$ (if not terminating earlier) and then executing  Phase 2 of   
Algorithm {\sc Gather(Cross,$\not$Chir)}.  
By Lemma \ref{lemmaC} we know that, if   termination did not occur by this round, 
then  the ring has been fully traversed by all agents. 

In case $k$ is known (but $n$ is not), 
each agent moves counterclockwise  terminating   
if the    $k$ agents are all at the same node.
As soon as it   passes by $k+1$ homebases, it discovers $n$.
At this point, it continues to move in the same direction switching to Phase 2 
at round $3n+1$ (unless gathering occurs before).
In fact,  by Lemma \ref{observationmovements}, 
we know that,  if an agent does not perform  $n$ steps in the first $3n$ rounds, then
all agents   are  in a single group and, knowing $k$,  they can immediately terminate.
 This means that after $3n$ rounds, if the agents have not terminated, they have however
 certainly performed a loop of the ring, they know $n$ (having seen $k+1$ home bases)
 and they    switched to   Phase $2$ by round $3n+1$.

\subsubsection{Algorithm {\sc Gather(Cross, Chir)}: {\sf Phase 2}}\label{ph2:knownchir}

When Phase 2 starts,  both $n$ and $k$ are known and Phase 2 of Algorithm {\sc Gather(Cross, Chir)} is identical to the one of 
  Algorithm {\sc Gather(Cross,$\not$Chir)}.

\medbreak

\noindent We then have:

 \begin{theorem}\label{lbl:th1}
  With chirality, cross detection and  knowledge of  either  $n$ or $k$,  \Gathering is solvable   in at most ${\cal O}(n)$ rounds from  any configuration $C \in {\cal C} \setminus {\cal P}$.
 $\qed$ 
 \end{theorem}


\section{Without Cross Detection}\label{sec:chiralitynocross}
 
 In this section we study the gathering problem  when there is no cross detection. 
 
 We focus first on the case when the absence of cross detection is
 mitigated by the presence of chirality. We show that 
gathering is possible in the same class of configurations as with cross detection,
albeit with a $O(n\log n)$  time complexity.
 
 We then examine  the most
 difficult case of absence of both cross detection and chirality. 
 We prove that in this case the class of feasible 
 configurations is smaller (i.e., cross detection is a computational
 separator). We show that gathering can be performed  from all
 feasible configuration in $O(n^2)$ time.

\subsection{Without Cross Detection: With Chirality}
\label{NoCrossChi}
The structure of the algorithm, {\sc Gather($\not$Cross,Chir)}, still follows the two Phases. However, 
when there is  chirality but no cross detection,  
 the difficulty lies  in the termination of Phase 2.

\subsubsection{Algorithm {\sc Gather($\not$Cross,Chir)}: Phase 1}
Notice that the Phase $1$ of Algorithm  {\sc Gather(Cross,Chir)} 
described in Section \ref{CrossChi} does not really make use of
cross detection.
So  the same Algorithm can be employed in this setting in both cases when   $n$ or $k$  are known.
Phase 1 terminates then   in ${\cal O}(n)$ rounds.
  
%

\subsubsection{Algorithm {\sc Gather($\not$Cross,Chir)}: Phase 2}\label{sec:chiralitynocrossp2}
Because of chirality, a leader node can be always elected, even when the initial configuration is in ${\cal E}$ (Property \ref{election}). We will show how to use this fact to modify  Phase $2$ of Algorithm  {\sc Gather(Cross,Chir)}  to work without assuming cross  detection.  
We will do so by designing a mechanism that will force the agents  
{\em never to cross each other}. The main consequence of this  fact is that, 
whenever two agents (or two groups of agents) would like to traverse the same edge  in opposite direction,
  only one of the two will be allowed to move thus ``merging" with the other. This mechanism is described below.

\paragraph{\bf  Basic no-crossing mechanism.} 
To avoid crossings,  each agent constructs an edge labeled bidirectional 
directed ring with $n$ nodes (called \augmented) and
it  moves on the actual ring according to the algorithm, but also  to  specific conditions 
dictated by the labels of the \augmented.
\begin{figure}[H]
  \centering
      \includegraphics[scale=0.6]{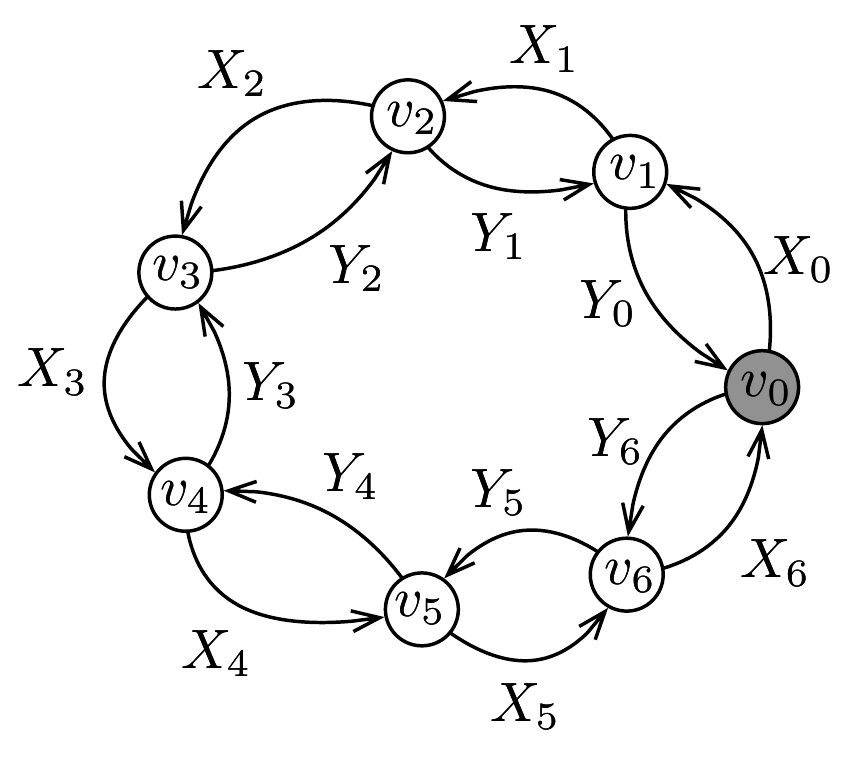}
  \caption{Example of the  \augmented   \label{fig:f1}} 
\end{figure}

In the  \augmented,  each edge  of the actual  ring  is replaced by two labeled oriented edges in the two directions. The label of each oriented edge $e_i$, $0 \leq i \leq n-1$, is either $X_i$ or $Y_i$ , where $X_i$ and $Y_i$ are infinite sets of integers.  
Labels $X_0 \ldots X_{n-1}$ are assigned to consecutive edges in counter-clockwise direction starting  from the leader node,
 while $Y_0 \ldots Y_{n-1}$  are assigned in   clockwise direction  (see Figure \ref{fig:f1}).

 Intuitively, we want to construct  these   sets of labels  in such a way that  
 $X_i$ and $Y_i$  have an empty intersection. In this way,
 the following meta-rule of movement will prevent any crossing:


\begin{quote}
{\em An agent is allowed to traverse an edge of the ring at round $r$ 
only if  $r$ is contained in the set of labels associated to the corresponding  oriented edge of the \augmented.}
\end{quote}

For this construction, we define    $X_i=\{s+m\cdot (2p+2)\mid (s\in S_i \vee s=2p), \forall m\in \mathbb N\}$,
 where $p=\lceil\log_2 n\rceil$, and $S_i$ is a subset of $\{0,1,\dots,2p-1\}$ of size exactly $p$ (note that  there are ${2p\choose p}\geq n$ possible choices for $S_i$).
 Indeed, there are $2^p = 2^{\lceil\log_2 n\rceil}\geq n$ ways to choose which elements of $\{0,1,\dots, p-1\}$ are in $S_i$; each of these choices can be completed to a set of size $p$ by choosing the remaining elements from the set $\{p,p+1,\dots,2p-1\}$. 
 Therefore there are at least $n$ available labels, and we can define the $X_i$'s so that they are all distinct. 
 Then we define  $Y_i$ to be the complement of $X_i$ for every $i$. That is, $X_i\cap Y_i=\varnothing$ and $X_i\cup Y_i=\mathbb N$.
 
  By construction, it follows that $|X_i\cap \{0,1,\dots,2p-1\}|=p$, and $|Y_i\cap \{0,1,\dots,2p-1\}|=p$, $\forall i$. 
  As a consequence, if $i\neq j$ and $m\in\mathbb N$, then $X_i$ and $Y_j$ have a non-empty intersection in $\{m,m+1,\dots,m+2p+1\}$. 
Furthermore, in  this labelling,   each $X_i$ contains all integers of the form $2p+m\cdot(2p+2)$, 
and each $Y_i$'s contains all integers of the form $2p+1+m\cdot(2p+2)$.

The following property is immediate by construction:
\begin{observation}
\label{obs:labelIntersection}
Let $m\in\mathbb N$ and let $I=\{m,m+1,\dots,m+2p+1\}$. Then, $X_i$ and $Y_j$ have a non-empty intersection in $I$ if and only if $i\neq j$, $X_i$ and $X_j$ have a non-empty intersection in $I$, and $Y_i$ and $Y_j$ have a non-empty intersection in $I$.
\end{observation}

From  the previous  observation, it follows that two agents moving following the \augmented   in opposite directions
 will never cross each other on an edge  of the actual  ring.

As a consequence of this fact, we can derive a bound on the number of rounds that guarantee two groups of robots moving in opposite direction, to ``merge". In the following lemma, we consider the execution of the algorithm proceeding in
 {\em periods}, where each period is composed by $2p+2$ rounds. We have:

\begin{lemma}\label{nodetect:obs1}
Let us consider two groups of agents, $G$ and $G'$, moving in opposite directions following the \augmented. After at most $n$ periods, that is at most ${\cal O}(n \log n)$ rounds, the groups will be at a distance $d \leq 1$ (in the direction of their movements). 
\end{lemma} 
\begin{proof}
Without loss of generality, let us assume that $G$ and $G'$ are initially positioned on two nodes, respectively $v$ and $v'$, trying to traverse two edges incident to  $v$ and $v'$.
If the two edges have labels that are the complement of each other  in the \augmented then, by construction, they are trying to traverse the same  edge in the actual ring  in opposite directions, and the lemma follows.

Let us then assume now that the two groups are trying to traverse edges whose  labels in the \augmented    are not the complement of each other. Since these  sets of labels  have a non empty intersection (Observation~\ref{obs:labelIntersection}), it follows that, in each period of $2p+2$ rounds, the adversary can block at most one of the two groups. Thus, there exists a round $r$ in which both groups try to cross two different edges, and at least one of them will succeed, hence moving of one step in the direction of the other group. Therefore, after at most $(n-1)(2p+2)$ rounds the two groups will be at a distance at most one in the directions of their movements. Since each period has ${\cal O}(\log n)$ rounds, the lemma follows. 
\end{proof}

 \begin{figure}[H]
\footnotesize

\begin{framed}
\begin{algorithmic}
\State States: \{\sReached,  \sReaching, {\sf ChangeDir}, {\sf ChangeState},  {\sf DirCommR}, {\sf DirCommS}, {\sf Term}\}.
\AtState{Phase 2}
 \If {$C \in {\cal P}$}
 \State {\sf unsolvable}()
    \State {\sf Go to State} {\sf Term}
 \EndIf
  \State resetAllVariables except \vTotOthers
    \State $dir={\sf leaderMinimumPath}$()
    \State \Explore({$dir\,\,|$ $\emph{seeElected} $: \sReached; $\vTtime=3n$: \sReaching  })
\AtState{\sReached}

    \If {\vTtime$\geq 3n$} 
    
        \State $dir={\sf clockwiseDirection}$()

    \State \Explore({$dir \,\,|$ $(\vBPtime \geq 4n+8 \,\, \lor  \,\, \vOthers=\vTotOthers)$: {\sf Term};})
    \EndIf
\AtState{\sReaching}
    \If {\vTtime$= 3n$}
    \State $dir={\sf counterclockwiseDirection}$()
     \EndIf
    \State \Explore({dir $|$ $(\vBPtime \geq 4n+8 \,\, \lor \,\, \vOthers=\vTotOthers)$: {\sf Term};})
  \END
\end{algorithmic}
\end{framed}	
\caption{{\sf Phase 2} of  Algorithm {\sc Gather($\not$Cross,Chir)} 
\label{algp2c1:gatheringnochirnocross}}
\end{figure}

We are now ready to describe the second Phase of the algorithm.
\paragraph{\bf {\sf Phase 2}.} 
In the following, when the agents are moving following the meta-rule in  the \augmented, we will use   variable $\vBPtime$,
 instead of $\vBtime$, indicating the number of consecutive periods in which the agent failed to traverse the current edge. 
 As in the case of  $\vBtime$, the new variable $\vBPtime$ is reset each time the agent traverses the edge, changes direction, or encounters new agents in its moving direction.

In the first $3n$ rounds,  each agent moves towards the elected node using the minimum distance path. After round $3n$, the agents move on the \augmented ring:  the group in state \sReached starts moving in clockwise direction, the group in state  \sReaching in counterclockwise. 
One of the two groups  terminates if $\vBPtime \geq n$ rounds or if $\vOthers=k$. This replaces the terminating condition $\vBtime=2n$ that was used in case of Cross detection.
Phase 2 of the Algorithm is shown in Figure \ref{algp2c1:gatheringnochirnocross}.
 
\begin{lemma}\label{lbl:lemmaChir}
 Phase 2 of    Algorithm {\sc Gather($\not$Cross,Chir)}  
  terminates in at most ${\cal O}(n  \log n)$ rounds, solving the \Gathering problem.
\end{lemma}
\begin{proof}
 Let us first prove that the algorithms terminates in ${\cal O}(n  \log n)$  rounds.
At the end of round $3n$ of Phase 2, we have at most one group of agents in state  \sReached and one group  in state \sReaching 
(Lemma \ref{observation2} derived in the case with cross-detection still holds).
 If there is only one of these group,  termination is immediate from condition  $\vOthers= k$.
If both groups are present (moving in opposite direction by construction) we have that, by Lemma \ref{nodetect:obs1} the two  groups will be at distance $1$ by at most round $3n+ n(2p+2)$, where $p$ is a quantity bounded by ${\cal O}(\log n)$. 
At this point,  they either meet in one node because only one of the two group will be allowed to cross the edge, and therefore they terminate by condition  $\vOthers= k$, or they are blocked by the adversary on two endpoints of the same edge. In this case, however,  they will terminate e  by condition $\vBPtime \geq n$. 
Notice  that,  if a group $G$ terminates by $\vBPtime \geq n$ gathering will be achieved, because by Lemma \ref{nodetect:obs1}, we have that the other group $G'$ is at the other endpoint of the edge where $G$ has been blocked. Therefore, $G'$ either terminates by condition $\vBPtime \geq n$, or it reaches the node where $G$ is and it terminates by condition $\vOthers= k$.
\end{proof}

From the previous Lemma, and the correctness of Phase 1 already discussed in Section \ref{CrossChi},
 the next theorem immediately follows. 

 \begin{theorem}\label{lbl:th2}
 With chirality  and  knowledge of $n$ or $k$,  \Gathering is solvable  from any configuration $C \in {\cal C} \setminus {\cal P}$. Moreover, there exists an algorithm solving \Gathering that terminates in ${\cal O}(n  \log n)$ rounds for any $C\in  {\cal C} \setminus {\cal P}$,  if $C\in {\cal P} $ the algorithm either solves \Gathering or it detects that the configuration is in ${\cal P}$. $\qed$
 \end{theorem}

\subsection{Without Cross Detection: Without Chirality}
\label{NoCrossNoChi}

In this section, we consider the most difficult setting when neither cross detection nor chirality are available.
We   show that in this case \Gathering 
 is impossible if $C\in {\cal E}$.  
 On the other hand, we provide a solution for  rings of known size from any initial configuration
$C\in {\cal C}\setminus ({\cal P}\cup{\cal E})$, which  works  in ${\cal O}(n^2)$ rounds. 
  We start this Section with the impossibility result.

\subsubsection{Impossibility for $C\in {\cal E}$\label{sec:impossE}}
\begin{theorem}\label{th:nocrossimposs}
Without chirality and without cross detection,  \Gathering is impossible when starting 
from a configuration $C \in {\cal E}$. This holds even if the agents   know $C$ (which implies 
knowledge of $n$ and $k$). 
\end{theorem}
\begin{proof}

\begin{figure}[H]
  \centering
      \includegraphics[scale=0.5]{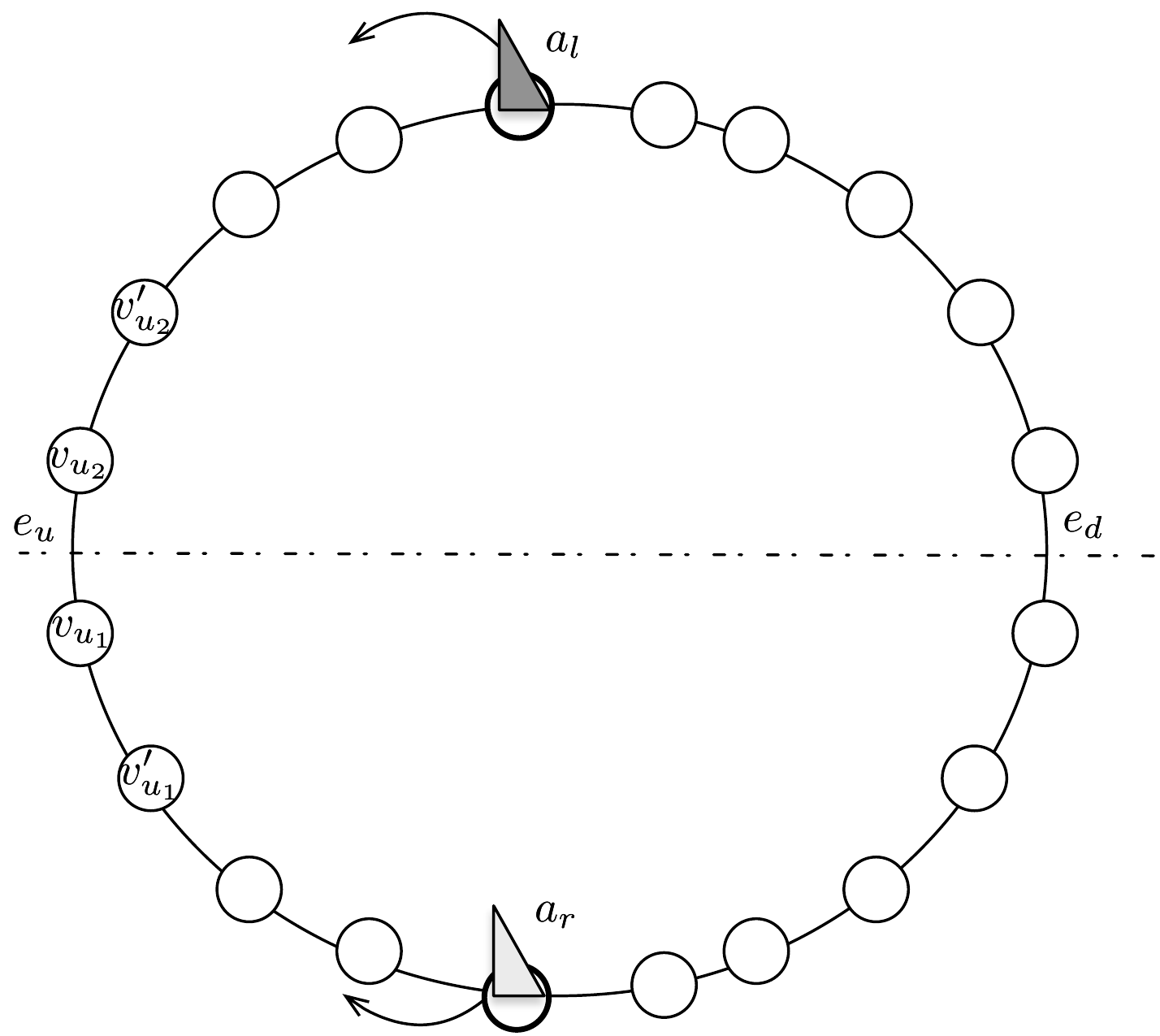}
  \caption{Configuration used to prove the impossibility of \Gathering when the configuration is in ${\cal E}$ and there is no cross detection\label{fig:f4}}
\end{figure}
By contradiction.
Consider an initial configuration $C$ 
with two agents $a_l,a_r$,   a unique axis of symmetry  passing through edges $e_{u},e_{d}$, and where the two homebases $h_l,h_r$ are at distance at least $4$ from $e_u$ and $5$ from $e_d$ (see an example in Figure \ref{fig:f4}). 
Let $A$ be an algorithm that solves gathering starting from configuration  $C$ 
in an execution $E$ where the adversary does not remove any edge. Note that,  because of symmetry,
without edge removals  the two agents  can 
cross each other only over $e_u$ or $e_d$  never meeting in the same node at the same time, so gathering could be achieved only on the two endpoint of one of these edges.
 Let us suppose, w.l.o.g,  that $A$ terminates when the two agents are on the endpoints   of edge $e_u = (v_{u_1},v_{u_2})$. 
Let $v_{u_2}'$ be  the  neighbour of $v_{u_2}$  different from $ v_{u_1}$ (resp. $v_{u_1}'$    the  neighbour of $v_{u_1}$  different from $ v_{u_2}$).
 Let $r_{f}$ be the   round in which $a_l$ reaches $v_{u_1}$ and terminates (note that $a_l$  could have passed by $v_{u_1}$ several time before, without terminating;  let $r_1$, possibly equal to $r_f$, be  the first round when $a_l$ reaches $v_{u_1}$). 
Agent  $a_l$ may reach   $v_{u_1}$ at round $r_{f}$ in two ways:
{\em Case 1)} after performing  a loop of the ring starting from $v_{u_1}$ 
(note that, during the loop, the agents may go back and forth over some nodes several times).
{\em Case 2)}  after  moving   in a certain direction for $X$ step and then  back for other $X$ step, possibly moving back and forth over some nodes several times.
In either case,  agent $a_r$ does exactly the symmetric moves of $a_l$ with respect to the symmetry  axis. 

Let us now   consider  an execution $E'$   starting from $C$  where the agents behave like in execution $E$ 
until they possibly find themselves blocked by an edge removal. 
We will show that the  edge removal  schedule chosen by the adversary
does not influence agent $a_l$, which behaves exactly as in execution $E$ terminating 
 in node $v_{u_1}$  at round $r_f$,  but  gathering is not achieved.   \\
%
No edge removal is done on the way of agent $a_l$  until  it  terminates  in node   $v_{u_1}$. 
If  $a_l$  does so by looping around the ring   ({\em Case 1}), also $a_r$  is performing an opposite loop and 
 the adversary blocks   $a_r$, on an endpoint of $e_d$, after the agents cross each other, for the last time,
 on $e_d$ during their loop. Regardless of the decision taken by $a_r$ at this point, 
when $a_l$ terminates, $a_r$ is at at least two  edges apart.
If $a_l$ is  reaching $v_{u_1}$ after moving for $X$ steps and coming back  ({\em Case 2}), $a_r$ is performing the symmetric moves 
and    the adversary behaves differently depending on various sub-cases.
{\em Case 2.1)} Assume first that $a_l$ (resp. $a_r$) leaves the set of nodes  $\{v_{u_1},v_{u_2}$, $v_{u_2}'\}$ 
(resp. $\{v_{u_2},v_{u_1}$, $v_{u_1}'\}$) at least once after round $r_1$.  In this case,  
if the agents do not traverse  $e_d$ or $e_u$,  then the adversary blocks $a_r$ the last time it leaves node $v_{u_2}$;
if instead  they  traverse  $e_d$, then $a_l,a_r$  cross on $e_d$ and   the adversary blocks $a_r$ on an endpoint of $e_d$ after their last cross.
Finally, if  the agents cross each other on $e_u$, then the adversary blocks $a_r$ as soon as it moves from $v_{u_1}'$.
In all these situations, when $a_l$ reaches $v_{u_1}$, $a_r$  is at least two edges apart.
{\em Case 2.2)} Assume now that $a_l$ never leaves the set of nodes  $\{v_{u_1},v_{u_2}$, $v_{u_2}'\}$ after round $r_1$. 
 The adversary   blocks agent $a_r$  right before it is entering for the first time in the set of nodes $\{v_{u_2},v_{u_1}$, $v_{u_1}'\}$,
 this would be undetectable by $a_l$, and, by construction,  $a_r$ would  be at distance at least $2$ from $v_{u_1}$, when $a_l$ terminates.
  
Being   run $E'$ undistinguishable for $a_l$ from the execution $E$, 
we have that, in $E'$,   $a_l$ terminates on $v_{u_1}$, while agent $a_r$ 
is not on a neighbour node of $v_{u_1}$.  
At this point the adversary blocks $a_r$  from any further move and gathering will never be achieved.
A contradiction.

\end{proof}


\subsubsection{Algorithm {\sc Gather($\not$Cross,$\not$Chir)}: {\sf Phase 1}}

As we know, the lack of cross detection is not a problem when there is a common chirality.
However, the  combination of   lack of both cross detection and chirality significantly complicates {\sf Phase 1}, and new mechanisms have to be devised to insure that all agents finish the ring exploration and correctly switch to Phase 2.
 
  In the following we will denote by $\vBtime'$ the value of \vBtime at the previous round, that is at round $\vTtime-1$. 

 \begin{figure}[H]
\footnotesize

\begin{framed}
\begin{algorithmic}
\State States: \{{\sf Init}, {\sf SyncR}, {\sf SyncL}, {\sf Term}\}.
\AtState{Init}
    \State \Explore({\dLeft$|$ \vTtime$ \geq (3n)(n+3) $: {\sf SyncL}; $\vBtime \geq (2n+2) \lor (\vBtime' \geq n+1 \land \pMeeting)$: {\sf Term}})
\AtState{SyncL}
    \State \Explore({\dLeft $|$ $ (\vTtime \geq (3n)(n+3)+2n+1 \land \vBtime >n) \lor \vOthers =\vTotOthers$:  {\sf Term};  $\vTtime \geq (3n)(n+3)+2n+1$: {\sf Phase 2}; $0 < \vBtime \leq n$: {\sf SyncR}})
\AtState{SyncR}
    \State \Explore({\dRight $|$ $  \vOthers =\vTotOthers$:  {\sf Term};  $\vTtime \geq (3n)(n+3)+2n+1$: {\sf Phase 2}; $\vBtime=1$: {\sf SyncL}})
\END
\end{algorithmic}
\end{framed}	
\caption{  {\sf Phase 1} of  Algorithm 
{\sc Gather($\not$Cross,$\not$Chir)}  
\label{algp1:gatheringnochirnocross}}
\end{figure}

%
%
%

Each agent  attempts to move along the ring in its own  left direction. 
An agent terminates in the {\sf Init} state if it has been blocked long enough ($\vBtime \geq 2n+2$), 
or if it was blocked for an appropriate amount of time and is now meeting a new agent ($\vBtime' \geq n+1 \land \pMeeting$). 
If an agent does not terminate by round $(3n)(n+3)$
it enters the {\em sync sub-phase} that lasts $2n$ rounds; this syntonization step  is  
used to ensure that, if a group of agents terminates in the {\sf Init} state  by condition  ($\vBtime \geq 2n+2$), all the remaining active agents will terminate correctly in this  sub-phase.

An agent with $\vBtime=0$ or $\vBtime> n $  starts  the {\em sync sub-phase} in state  {\sf SyncL}.
Instead, an agent with $0 < \vBtime \leq n$  starts in state  {\sf SyncR} and resets $\vBtime$ to zero.    
%
    In the successive rounds, when/if  an agent is blocked  in state  {\sf SyncR} (resp.  {\sf SyncL})  
    it   switches direction, changes  state to  {\sf SyncL} (resp.  {\sf SyncR}),
     resetting the variable   $\vBtime$ to 0.
 The agent terminates if it either detects $k$ agents at its current node, or if it never moved ($\vEsteps = 0$)
by round $(3n)(n+3)+2n+1$. Otherwise,  at  round $(3n)(n+3)+2n+1$ it starts {\sf Phase 2}.

\begin{observation}
\label{obs:Ph1}
If an edge is missing for $3n$ consecutive rounds, between rounds $0$ and $(3n)(n+3)$, then all agents terminate. Therefore, if an agent has not terminated by round $(3n)(n+3)$, then it has done a complete tour of the ring. 
\end{observation}

\begin{lemma}\label{lemma:p1nochirnocross}
If an agent does not terminate at the end of {\sf Phase 1}, then no agent terminates and all of them have done at least one complete loop of the ring. If an agent terminates during {\sf Phase 1}, then all agents terminate and \Gathering  is correctly solved.
\end{lemma}
\begin{proof}
The proof proceeds by considering all possible cases when an agent $a^*$ can terminate during {\sf Phase 1}. 

If $a^*$ terminates at a round $r \leq (3n)(n+3)$, then it is blocked on a missing edge, say at node $v$. 
Also, by definition of state {\sf Init}, either condition $\vBtime \geq 2n+2$ or $\vBtime' \geq n+1 \land \pMeeting$ is satisfied by $a^*$ at round $r$. 

\begin{itemize}
\item If $\vBtime \geq 2n+2$ is verified at round $r$, then all agents with the same direction of movement of $a^*$ are terminated as well at $v$, and for all of them $\vBtime \geq 2n+2$ is verified. Let us consider the agents with direction of movement opposite to that of $a^*$. If there is no such agent, then the  lemma clearly  follows. Otherwise, they form a group, call it $G$, on the other endpoint of the missing edge. Note that, at round $r$, the agents in $G$ have been blocked for at least $2n+2 - (n+1)=n+3$ rounds, hence, at round $r$, for the agents in $G$, (**) $\vBtime \geq n+3$ is verified.

If the agents in $G$ are terminated at round $r$, then the lemma follows. Otherwise, at round $r$, for the agents in $G$, $\vBtime' \geq n+1$ is satisfied (see (**)). If the agents in $G$ do not change state, and if the edge is missing for the next $n+1$ rounds, then agents in $G$ terminate on condition $\vBtime \geq 2n+2$. Otherwise, if the edge comes alive within the next $n+1$ rounds, the agents in $G$ will cross it, meet the (terminated) agents in $v$, and terminate as well. Thus, gathering is correctly achieved, and the lemma follows. On the other hand, if the agents in $G$ switch to the {\sf SyncL} state, two cases can occur: $(a)$ the edge is missing for the next $2n$ rounds: in this case, the agents in $G$ terminate on condition $\vTtime \geq (3n)(n+3)+2n+1 \land \vBtime >n$; $(b)$ the edge comes alive within the next $2n$ rounds: in this case, the agents in $G$ cross it and meet all the other (terminated) agents in $v$. In both cases, the gathering is correctly achieved, and the lemma follows.
 
\item If $\vBtime' \geq n+1 \land \pMeeting$ is verified at round $r$, then all agents with the same direction of movement of $a^*$ are terminated as well at $v$, and for all of them $\vBtime' \geq n+1 \land \pMeeting$ is verified. Let us consider the agents with direction of movement opposite to that of $a^*$. They form a group, call it $G$, on the other endpoint of the missing edge.

First note that, since \pMeeting is verified at round $r$, at the previous round $r-1$, $a^*$ was at $v$'s previous node (according to the chirality of $a^*$), say $v_{i-1}$. 
Moreover, by $\vBtime' \geq n+1$ the edge between $v_{i-1}$ and $v$ is missing at round $r-1$, and the agents in $G$ must have entered in a terminal state by round $r-1$ (Notice  that the agents in $G$ had enough time to reach the other endpoint and enter   the port, therefore if they were not in terminal state at round $r$, then a cross would have occurred at round $r$, hence \pMeeting not satisfied); also, the agents in $G$ terminated on condition $\vBtime \geq 2n+2$. Therefore, at round $r$, gathering is correctly achieved at $v$, and the lemma follows.
\end{itemize}
It remains to prove the correctness of the termination of $a^*$, say at node $v$, in a round $r > (3n)(n+3)$, that is during the sync sub-phase. In this case, by Observation~\ref{obs:Ph1}, all agents know $k$. The correctness of the termination when condition $\vOthers=k$ holds is trivial; thus, let us consider the case when termination occurs because $\vTtime \geq (3n)(n+3)+2n+1 \land \vBtime >n$ holds.

Let $G$ be the group of agents that terminates for this condition. 
Notice  that  this  group must have never left  the {\sf SyncL} state.
In fact, any agent that enters state   {\sf SyncR} at any time
during the sub-phase will never have  $\vBtime   > n$ for the rest of the sub-phase 
even when it becomes  {\sf SyncL}; thus,
the only agents with $\vBtime   > n$  at time $ (3n)(n+3)+2n+1$
are those that entered the
  {\sf SyncL} state with   $\vBtime   > n$ and never switched.
This implies that the edge on which $G$ terminates has been  missing  for the whole execution of the sub-phase.

At round  $(3n)(n+3)+n$ all agents with direction of movements opposite to the one of the agents in $G$ are in another group $G'$ on the other endpoint of the missing edge.
If the agents in $G'$ are already terminated, the termination of the agents in $G$ correctly solves gathering, and the lemma follows. 
 If instead the agents in $G'$ are not terminated, but they are also in state  {\sf SyncL} with $\vBtime >n$  then, they will also terminate for $\vTtime \geq (3n)(n+3)+2n+1 \land \vBtime >n$, therefore gathering is correctly achieved, and the lemma follows. Otherwise, the agents are either in state {\sf SyncL}   or  {\sf SyncR}  with  $\vBtime  \leq n$, therefore they will change direction at round $(3n)(n+3)+n+1$ in the next $n$ rounds they will  move towards $G$, and will reach group $G$; when this happens, gathering will be correctly achieved (on condition $\vOthers=k$), and the lemma follows.

\end{proof}

\subsubsection{Algorithm {\sc Gather($\not$Cross,$\not$Chir)}: {\sf Phase 2}}

By Lemma~\ref{lemma:p1nochirnocross}, at the end of {\sf Phase 1} each agent knows the current configuration. Since we know that the problem is not solvable
for  initial configurations   ${{\cal C} \in \cal E}$ (Theorem \ref{th:nocrossimposs}), the initial configuration must be 
non-symmetric (i.e., without any axis of symmetry) or symmetric but with the unique axis of symmetry   going through a node. 
In both cases, the agents can agree on a common chirality. In fact,  if ${\cal C}$ does not have any symmetry axes, the agents  can agree, for example, on the direction of the lexicographically smallest sequence of homebases inter distances.
If instead there is
an  axis of symmetry   going through a node  $v_L$, they can agree on the direction of  the port of $v_L$
with the  smallest label.

We can then use as Phase 2, the one of Algorithm    {\sc Gather($\not$Cross,Chir)}   presented in   Section~\ref{sec:chiralitynocrossp2}.

\begin{theorem} 
\label{NoChirNoCross}
Without chirality, \Gathering is solvable in rings of known size without cross detection from all  $C\in {\cal C}\setminus({\cal P} \cup {\cal E})$. Moreover, there exists an algorithm solving \Gathering that terminates in ${\mathcal O(n^2)}$ rounds for any $C\in  {\cal C}\setminus({\cal P} \cup {\cal E})$,  if $C\in {\cal P} \cup {\cal E}$ the algorithm either solves \Gathering or it detects that the configuration is in ${\cal P} \cup {\cal E}$. 
 \end{theorem}
\begin{proof}
By Lemma \ref{lemma:p1nochirnocross}, it follows the correctness and the ${\cal O}(n^2)$ bound of {\sf Phase 1}. The correctness and complexity of {\sf Phase 2}, follows by the Lemma \ref{lbl:lemmaChir} of Section \ref{sec:chiralitynocrossp2}. The last statement of the theorem is obvious by Lemma \ref{lemma:p1nochirnocross}, if at the end of {\sf Phase 1} the gathering is not solved, then agents know ${\cal C}$, therefore they can detect if the configuration is in  ${\cal P} \cup {\cal E}$.
\end{proof}

 \section{Conclusion}
  
  In this paper we have investigated the problem of \Gathering in a dynamic rings. When $n$ is known, we presented a complete characterisation on the initial configurations where \Gathering is solvable, with and without chirality and with and without the capability to detect agents crossing. Interestingly, in such dynamic setting the knowledge of $n$ cannot be trade-off with the knowledge of $k$, this is in contrast with the known results for \Gathering in static rings. An open problem is to investigate the complexity gap between the algorithms that solve \Gathering with cross detection and the algorithms that do not use cross detection. Our non-crossing technique introduces a complexity of ${\cal O}(n \log n)$ rounds, it would be interesting to show if such $\log n$ factor is necessary or not. 
  \color{black}

\newpage


\begin{thebibliography}{99}


\bibitem{AlG03}
S. Alpern and S. Gal.
\newblock {\em The Theory of Search Games and Rendezvous}.
\newblock  Kluwer, 2003.



 \bibitem{BarFFS03}
 L.~Barri\`ere, P.~Flocchini, P.~Fraigniaud, and N.~Santoro.
 \newblock Rendezvous and election of mobile agents: Impact of sense of direction.
\newblock  {\em Theory of Computing Systems} 40(2): 143--162, 2007. 


\bibitem{BasG91}
V. Baston and S. Gal. 
\newblock Rendezvous search when marks are left at the starting points. 
\newblock {\em Naval Research Logistics} 38: 469--494, 1991.

\bibitem{BiRSSW15}
M.~Biely, P.~Robinson, U.~Schmid, M. Schwarz, and K. Winkler.
\newblock Gracefully degrading consensus and k-set agreement in directed dynamic networks.
\newblock  {\em 2nd International Conference on Networked Systems} (NETSYS), 109-124, 2015.
  
\bibitem{Marge}
M.~Bournat, A.~Datta, and S.~Dubois.
\newblock Self-stabilizing robots in highly dynamic environments.
\newblock  {\em 18th International Symposium on Stabilization, Safety, and Security of Distributed System} (SSS), 54-69, 2016.

  \bibitem{BoDD16}
S.  Bouchard, Y. Dieudonne, and B. Ducourthial.
\newblock     Byzantine gathering in networks.
\newblock     {\em Distributed Computing}  29(6): 435--457, 2016.


  \bibitem{CaFMS14}
A.~Casteigts, P.~Flocchini, B.~Mans, and N.~Santoro.
\newblock Measuring temporal lags in
  delay-tolerant networks.
  \newblock {\em IEEE Transactions on Computers} 63~(2): 
  397--410, 2014.
 
  
  \bibitem{CaFQS12}
A.~Casteigts, P.~Flocchini, W.~Quattrociocchi, and N.~Santoro.
\newblock  Time-varying graphs  and dynamic networks.
  \newblock {\em Int. Journal of Parallel, Emergent and Distributed  Systems} 27~(5): 387--408, 2012.
  
\bibitem{ChDS07} 
J. Chalopin, S. Das, and N. Santoro.
\newblock Rendezvous of mobile agents in unknown graphs with faulty links.
\newblock    {\em 21st International Symposium  on Distributed Computing} (DISC), 108--122, 2007.

\bibitem{ChDLP14}   
J. Chalopin, Y. Dieudonne, A. Labourel, and A. Pelc.
\newblock Fault-tolerant rendezvous in networks.
\newblock {\em 41st International Colloquium on Automata, Languages, and Programming}  (ICALP), 411--422, 2014.

\bibitem{CiFPS12}
M.~Cieliebak, P.~Flocchini, G.~Prencipe, and N.~Santoro.
\newblock Distributed computing by mobile robots: Gathering.
\newblock {\em SIAM Journal on Computing} 41(4): 829--879, 2012.




\bibitem{CoP04}
R.~Cohen  and D.~Peleg.
\newblock Convergence properties of the gravitational algorithm in asynchronous  robot systems.
\newblock {\em SIAM Journal on Computing} 34: 1516--1528, 2005.

\bibitem{CzLP12}
J.~Czyzowicz, A.~Labourel, and A.~Pelc.
\newblock How to meet asynchronously (almost) everywhere.
\newblock {\em ACM Transactions on Algorithms}  8(4): 37:1--37:14, 2012.

\bibitem{CzDKK08}
J.~Czyzowicz, S. Dobrev, E. Kranakis, and D. Krizanc.
\newblock The power of tokens: Rendezvous and symmetry detection for two mobile agents in a ring.
\newblock {\em 34th Conference on Current Trends in Theory and Practice of Computer Science} 
(SOFSEM), 234--246, 2008.


\bibitem{DasLMMS16}
S. Das, F.L. Luccio, R. Focardi, E. Markou, D. Moro, and M. Squarcina.
 \newblock Gathering of robots in a ring with mobile faults.
 \newblock {\em 17th Italian Conference on Theoretical Computer Science}
 (ICTCS),
122--135, 2016.


\bibitem{DasLM15}
S. Das, F.L. Luccio, and E. Markou.
 \newblock Mobile agents rendezvous in spite of a malicious agent.
 \newblock {\em 11th International Symposium on Algorithms and Experiments for Sensor Systems, Wireless Networks and Distributed Robotics}
 (ALGOSENSORS),
211--224, 2015.

\bibitem{Deg+2011}
B.~Degener, B.~Kempkes, T.~Langner, F.~Meyer auf~der Heide, P.~Pietrzyk, and  R.~Wattenhofer.
\newblock A tight runtime bound for synchronous gathering of autonomous robots  with limited visibility.
\newblock  {\em 23rd ACM Symposium on Parallelism in   Algorithms and Architectures} (SPAA),   139--148, 2011.

\bibitem{DemGKKPV06}
G. De Marco, L. Gargano, E. Kranakis, D. Krizanc, A. Pelc, and U. Vaccaro.
\newblock Asynchronous deterministic rendezvous in graphs.
\newblock {\em Theoretical Computer Science} 355:  315--326, 2006.
    
\bibitem{DesFP03}
A. Dessmark, P. Fraigniaud, and A. Pelc. 
\newblock Deterministic rendezvous in graphs.
\newblock  {\em European Symposium on Algorithms} (ESA), 184--195, 2003.

 \bibitem{DieP16}     
Y. Dieudonn\'{e} and A. Pelc.
\newblock Anonymous meeting in networks. 
\newblock {\em Algorithmica} 74(2), 908--946, 2016.

\bibitem{DiePV15} 
Y. Dieudonn\'{e}, A. Pelc, and V. Villain.
\newblock How to meet asynchronously at polynomial cost.
\newblock {\em SIAM Journal on Computing} 44(3):844--867, 2015.

\bibitem{DilDFS16} 
G.A. Di Luna, S. Dobrev, P. Flocchini, and N. Santoro.
\newblock Live exploration of dynamic rings.
\newblock {\em 36th IEEE International Conference on Distributed Computing Systems}, (ICDCS), 570-579, 2016. 

\bibitem{DilFPPSV17} 
G.A. Di Luna, P. Flocchini, L. Pagli, G. Prencipe,  N. Santoro, and G. Viglietta.
\newblock Gathering in dynamic rings.
\newblock {\em Arxiv}, April,  2017.


\bibitem{DoFPS03}
S. Dobrev, P.~Flocchini, G.~Prencipe, and N.~Santoro. 
\newblock  Multiple agents rendezvous in a ring in spite of a black hole.   
\newblock {\em 7th International Conference on Principles of Distributed Systems} (OPODIS),  34--46, 2003.
  
 \bibitem{FlKKSS04}
P. Flocchini, E. Kranakis, D. Krizanc, N. Santoro, and C. Sawchuk.
\newblock Multiple mobile agent rendezvous in the ring. 
\newblock  {\em 6th Latin American Conference on Theoretical Informatics} (LATIN), 599--608, 2004.

\bibitem{FlMS13}
P.~Flocchini, B.~Mans, N.~Santoro.
\newblock  On the exploration of time-varying networks.
\newblock  {\em Theoretical Computer Science} 469: 53--68, 2013.
  

\bibitem{FlPSW05}
P.~Flocchini, G.~Prencipe, N.~Santoro, and P.~Widmayer.
\newblock Gathering of asynchronous robots with limited visibility.
\newblock {\em Theoretical Computer Science} 337(1-3):147--168, 2005.

\bibitem{FlSVY16} 
P. Flocchini, N. Santoro, G. Viglietta, and M. Yamashita.
\newblock Rendezvous with constant memory.
\newblock {\em Theoretical Computer Science}  621:57--72, 2016.

\bibitem{HaK12} 
B. Haeupler and F. Kuhn.
\newblock Lower bounds on information dissemination in dynamic networks.
\newblock {\em 26th International Symposium on Distributed Computing}  (DISC),  166--180, 2012.

\bibitem{IlKW14}
D.~Ilcinkas, R.~Klasing, and A.M. Wade.
\newblock Exploration of constantly connected dynamic graphs based on cactuses.
\newblock In {\em Proc.  21st Int. Coll. Structural
  Inform. and Comm. Complexity (SIROCCO)},  250--262, 2014.

\bibitem{IlW13}
D.~Ilcinkas and A.M.~Wade.
\newblock Exploration of the T-Interval-connected dynamic graphs: the case of the ring.
\newblock In {\em Proc. 20th Int. Coll. on Structural
  Inform. and Comm. Complexity (SIROCCO)},  13-23, 2013.


 
\bibitem{KlMP08}
R Klasing, E Markou, and A Pelc.
\newblock Gathering asynchronous oblivious mobile robots in a ring.
\newblock {\em Theoretical Computer Science} 390 (1): 27--39, 2008.

\bibitem{KoM08} 
D.R. Kowalski and A. Malinowski.
\newblock How to meet in anonymous network. 
\newblock {\em Theoretical Computer Science} 399(1Ð2): 141--156, 2008.

\bibitem{KoP04}
D.R. Kowalski and A. Pelc.
\newblock Polynomial deterministic rendezvous in arbitrary graphs.
\newblock {\em 15th Annual Symposium on Algorithms and Computation} (ISAAC), 644--656, 2004.
    
 \bibitem{KrKM06}
E. Kranakis, D. Krizanc, and E. Markou.
\newblock  Mobile agent rendezvous  in a synchronous torus. 
\newblock   {\em 7th Latin American Conference on Theoretical Informatics} (LATIN), 653--664, 2006.
   
\bibitem{KrKM10}
E. Kranakis, D. Krizanc, and E. Markou.
\newblock {\em The Mobile Agent Rendezvous Problem in the Ring}. 
\newblock Morgan \& Claypool, 2010.

\bibitem{KrKSS03}
E. Kranakis, D. Krizanc, N. Santoro, and C. Sawchuk. 
\newblock Mobile agent rendezvous problem in the ring.
\newblock {\em 23rd International Conference on Distributed Computing Systems} (ICDCS),  592-599, 2003.


\bibitem{KuLO10}
F.~Kuhn, N.~Lynch, R.~Oshman.
\newblock  Distributed computation in dynamic networks.
\newblock   {\em 42th Symposium on Theory of Computing} (STOC),  513--522, 2010.
  
  \bibitem{KuMO11}
F.~Kuhn, Y.~Moses, R.~Oshman.
\newblock Coordinated consensus in dynamic networks.
\newblock   {\em 30th Symposium on Principles of Distributed Computing}  (PODC), 1--10, 2011.

  
    \bibitem{KuO11}
F.~Kuhn and R.~Oshman.
\newblock Dynamic networks: Models and algorithms.
\newblock   {\em SIGACT News} 42(1): 82--96, 2011.


\bibitem{LiMA07}
J.~Lin, A.S. Morse, and B.D.O. Anderson.
\newblock The multi-agent rendezvous problem. Parts 1 and 2.
\newblock {\em SIAM Journal on Control and Optimization}, 46(6): 2096--2147,  2007.
  
  
\bibitem{MiP16}   
A. Miller and A. Pelc.
\newblock Time versus cost tradeoffs for deterministic rendezvous in networks.
\newblock {\em Distributed Computing}   29(1): 51Ð64, 2016.

 \bibitem{PaPV15}
L. Pagli, G. Prencipe, and G. Viglietta.
\newblock Getting close without touching: Near-gathering for autonomous mobile robots.
\newblock {\em Distributed Computing} 28(5):333--349, 2015.

\bibitem{Pe12} 
A. Pelc.
\newblock Deterministic rendezvous in networks: A comprehensive survey.
\newblock {\em Networks}  59(3):  331-347, 2012.

\bibitem{Sawchuk04}
C. Sawchuk. 
\newblock {\em  Mobile Agent Rendezvous in the Ring}. 
\newblock Ph.D Thesis, Carleton University, January 2004.


\bibitem{TaZ14} 
A. Ta-Shma and U. Zwick.
\newblock Deterministic rendezvous, treasure hunts, and strongly universal exploration sequences. 
\newblock {\em ACM Transactions on Algorithms} 10(3), 2014.

\bibitem{YuY96}
X. Yu and M. Yung. 
\newblock Agent rendezvous: A dynamic symmetry-breaking problem. 
\newblock  {\em International Colloquium on Automata, Languages, and Programming} (ICALP),   610-621, 1996.


\end{thebibliography}
\end{document}